\documentclass[english,prodmode,acmtalg]{acmsmall}

\usepackage[ruled]{algorithm2e}

\SetAlFnt{\small}
\SetAlCapFnt{\small}
\SetAlCapNameFnt{\small}
\SetAlCapHSkip{0pt}
\IncMargin{-\parindent}

\acmVolume{V}
\acmNumber{N}
\acmArticle{A}
\acmYear{2016} 
\acmMonth{12} 

\setcopyright{acmcopyright}

\doi{DOI}

\issn{1549-6325/2015}

\usepackage{versions}
\includeversion{cuts}

\usepackage{csquotes}
\MakeOuterQuote{"}

\usepackage{url}
\newtheorem{theorem*}{Theorem}
\newtheorem{lemma*}{Lemma}
\newtheorem{corollary*}{Corollary}
\newdef{assumption}[theorem]{Assumption}

\usepackage{graphicx}
\usepackage{breakcites}
\usepackage{dirtree}

\usepackage{float}
\floatstyle{ruled}
\newfloat{algorithm}{tbp}{loa}
\floatname{algorithm}{Algorithm}

\usepackage{multirow}
\usepackage{array}

\usepackage{pstricks,pst-node}

\usepackage[shortcuts]{extdash}

\usepackage{caption}
\usepackage{framed}
\usepackage{amsmath}

\def\envyfree(#1,#2){\ensuremath{\textrm{EnvyFree}[#1,#2]}}
\def\envyfreevip(#1,#2){\ensuremath{\textrm{EnvyFreeVIP}[#1,#2]}}
\def\factor{\big\lceil \frac{M\ln{(1/\epsilon)}}{n}\big\rceil}
\def\equalize(#1){\ensuremath{\textrm{Equalize}(#1)}}

\newcommand{\piece}[1]{\ensuremath{\widehat{#1}}}
\newcommand{\piececut}[2]{\ensuremath{\widehat{#1}_{#2}}}
\newcommand{\agent}[1]{\ensuremath{A_{#1}}}
\newcommand{\agentspiece}[2]{\ensuremath{\piece{#2}_{#1}}}

\newcommand{\agentnode}[1]{\circlenode[linecolor=white]{#1}{\textcolor{blue}{#1}}} 
\newcommand{\piecenode}[1]{\dianode[linecolor=white]{#1}{\textcolor{blue}{\piece{#1}}}}
\newcommand{\piecenodecut}[2]{\dianode[linecolor=white]{#1#2}{\textcolor{blue}{\piececut{#1}{#2}}}}
\newcommand{\like}[2]{\ncline[linecolor=green!50!black!50]{->}{#1}{#2}} 
 
\newcommand{\likemaybe}[2]{\ncline[linecolor=gray,linestyle=dotted]{->}{#1}{#2}} 

\newcommand{\threeagents}{
	\agentnode{A} & \agentnode{B} & \agentnode{C} \\[12mm]
	\piecenode{1} & \piecenode{2} & \piecenode{3} 
}

\newcommand{\threeagentsdomination}{
	& \agentnode{A} &   \\
	\agentnode{B} &               & \agentnode{C}
}

\newcommand{\fouragents}{
	\agentnode{A} & \agentnode{B} & \agentnode{C} & \agentnode{D} \\[12mm]
}

\newcommand{\fouragentsdomination}{
	& \agentnode{A} &   \\
	\agentnode{B} &               & \agentnode{C} \\
	& \agentnode{D} &
}

\newcommand{\fiveagentsdomination}{
	& \agentnode{B} &   \\
	\agentnode{D} &               & \agentnode{E} \\
	& \agentnode{C} & \agentnode{A}
}

\newcommand{\agentcuts}[3]{
	\psline(0,5)(80,5)
	\rput[t](10,4){A}
	\rput[t](20,4){B}
	\rput[t](30,4){C}
	\rput[t](50,4){#1}
	\rput[t](60,4){#2}
	\rput[t](70,4){#3}
}

\newcommand{\pieceframe}[3]{
	\psframe[linecolor=black!30,fillstyle=solid,fillcolor=yellow!30,hatchangle=90,hatchcolor=yellow](#2,0)(#3,10)
	\rput[b](! #2 #3 add 2 div 6){#1}
}

\newcommand{\agentspieceframe}[4]{
	\pieceframe{\agentspiece{#1}{#2}}{#3}{#4}
}

\newcommand{\range}[2]{\in\{#1,\dots,#2\}}

\newcommand{\citet}[1]{\citeN{#1}}
\newcommand{\citep}[1]{\cite{#1}}

\begin{document}
	
\markboth{Erel Segal-Halevi et al.}{Bounded-Time Envy-Free Cake-Cutting with Free Disposal}

\title{Waste Makes Haste: Bounded Time algorithms for Envy-Free Cake Cutting with Free Disposal}

\author{
	EREL SEGAL-HALEVI and AVINATAN HASSIDIM and YONATAN AUMANN
	\affil{\\ Bar-Ilan University, Ramat-Gan 5290002, Israel}
}

\begin{abstract}
We consider the classic problem of envy-free division of a heterogeneous good ("cake") among several agents. It is known that, when the allotted pieces must be connected, the problem cannot be solved by a finite algorithm for 3 or more agents. The impossibility result, however, assumes that the entire cake must be allocated. In this paper we replace the entire-allocation requirement with a weaker \emph{partial-proportionality} requirement: the piece given to each agent must be worth for it at least a certain positive fraction of the entire cake value. We prove that this version of the problem is solvable in bounded time even when the pieces must be connected. We present simple, bounded-time envy-free cake-cutting algorithms for: (1) giving each of $n$ agents a connected piece with a positive value; (2) giving each of 3 agents a connected piece worth at least 1/3; (3) giving each of 4 agents a connected piece worth at least 1/7; (4) giving each of 4 agents a disconnected piece worth at least 1/4; (5) giving each of $n$ agents a disconnected piece worth at least $(1-\epsilon)/n$ for any positive $\epsilon$.
\end{abstract}

%
%

%
%
 

\terms{Algorithms, Economics}


\acmformat{Erel Segal-Halevi, Avinatan Hassidim, and Yonatan Aumann, 2016. Waste Makes Haste: Bounded Time algorithms for Envy-Free Cake Cutting with Free Disposal.}

\begin{bottomstuff}
A preliminary version of this paper appeared in the proceedings of AAMAS 2015 \citep{SegalHalevi2015Waste}. In the present paper, the algorithms have been greatly shortened and simplified. Algorithms 3, 4 and 5 are new.

This paper is supported in part by the Doctoral Fellowships of Excellence Program, Mordecai and Monique Katz Graduate Fellowship Program, ISF grant 1083/13, ISF grant 1241/12 and BSF grant 2012344.
\end{bottomstuff}
\maketitle

\pagebreak{}

\section{Introduction}
Fair cake-cutting is an active field of research with applications in mathematics, economics, and recently also in AI \citep{Procaccia2015Cake}. The basic setting considers a heterogeneous good, usually described as a one-dimensional interval, that should be divided among several agents. The different agents may have different preferences over the possible pieces of the good. The goal is to divide the good among the agents in a way that is deemed ``fair''.
Fairness can be defined in several ways, of which \emph{proportionality} and \emph{envy-freeness} are the most commonly used.

\emph{Proportionality} means that each agent gets at least its ``fair-share'' of the good, i.e.\ with $n$ agents, the piece allotted to each agent is worth for him at least $1/n$ of the value of the entire good.  \emph{Envy-freeness} means that every agent believes that its piece is weakly better than any other piece --- no agent would prefer to get a piece allotted to another agent. An additional requirement in cake-cutting, particularly relevant when the divided resource is land, is \emph{connectivity} --- each agent must be given a single contiguous piece.

Proportional division is a relatively easy task, and the initial work of \citet{Steinhaus1948Problem} already provided an algorithm for $n$ agents with connected pieces. The algorithm works in a \emph{query model}. It asks the agents queries of two types: "what is the value of piece $X$ for you?" and "what piece is worth a fraction $r$ of the cake for you?", and proceeds according to their replies (the query model has been formalized later by \citet{Robertson1998CakeCutting}). The number of queries required by Steinhaus' algorithm is polynomial in the number of agents. 

Envy-free division, on the other hand, turns out to be much more challenging. With connected pieces, the only algorithm for envy-free division is an infinite one; that is, it may require an infinite number of queries to reach an envy-free division \citep{Su1999Rental}. Indeed, \citet{Stromquist2008Envyfree} proved that this is necessarily so; any algorithm for computing an envy-free division with connected pieces must require an infinite number of queries on some inputs. This is so even when there are only 3 agents!

A closer examination of these discouraging result reveals that it critically relies on the assumption that \emph{the entire cake must be divided}. In many practical situations, it may be possible to leave some parts of the cake undivided, a possibility termed \emph{free disposal}. 
If, for example, your children spend too much time quarreling over the single cherry on top of the cake, one practical solution is to eat that cherry yourself and divide only the rest of the cake. As another example, when dividing land it is usually possible (and sometimes even desirable) to leave some parts of the land unallocated, so that they can be used freely by the public. The question of interest in this paper is thus: 
\begin{quote}
If free disposal is allowed, can an envy-free allocation be computed using a bounded number of queries?
\end{quote}
This question, however, turns out to have a trivial, but uninteresting, answer;  It is always possible to give nothing to all agents, which is an envy-free allocation.  Thus, the interesting question is whether it is possible to devise a bounded-time algorithm for envy-free division in which each agent gets a \emph{strictly positive} value.

\subsection{Results}

Our first algorithm provides an affirmative answer to this question:

\begin{theorem*}[Section \ref{sec:n-agents}]
\label{thm:n-agents} 
Assuming free disposal, it is possible to divide a cake envy-freely among $n$ agents, giving each agent a connected piece with a value of at least $1/2^{n-1}$, using 
\begin{cuts}
$2^{n-1}-1$ cuts and
\end{cuts}
$O(n\cdot 2^n)$ queries.
\end{theorem*}

Having established that bounded-time algorithms indeed exist, we next consider the \emph{quality} of the solution they offer. We measure the quality of an algorithm by the value guarantee it gives to each agent. We say that an allocation is \textbf{$\envyfree(n,M)$} if it is both envy-free and gives each of $n$ agents at least $1/M$ of the total cake value (a formal definition is given in Section \ref{sec:The-Model}). The algorithm of Theorem \ref{thm:n-agents} finds $\envyfree(n,2^{n-1})$ allocations. Ideally, we would like an $\envyfree(n,n)$ allocation, which is both envy-free and proportional and thus satisfies the two most common fairness criteria. Such an allocation guarantees each agent at least $1/n$ of the total cake value, which is the largest fraction that can be guaranteed (for example, if all agents have identical valuations, it is impossible to give all of them more than $1/n$).

Our next result shows that the ideal goal of an envy-free and proportional division with connected pieces can be attained in bounded time for three agents (Section \ref{sec:3-agents}):

\begin{theorem*}\label{thm:3-agents}
Assuming free disposal, it is possible to find 
an $\envyfree(3,3)$ allocation with connected pieces 
using at most 
	\begin{cuts}
	3 cuts and 
	\end{cuts}
54 queries.
\end{theorem*}

For 4 agents, we do not have an envy-free-proportional algorithm with connected pieces. Our best result so far guarantees each agent 1/7 of the total cake value (Section \ref{sec:4-agents}):
\begin{theorem*}\label{thm:4-agents}
Assuming free disposal, it is possible to find 
an $\envyfree(4,7)$ allocation with connected pieces 
using at most 
	\begin{cuts}
	6 cuts and 
	\end{cuts}
65 queries.
\end{theorem*}
While this algorithm does not solve the envy-free-proportional problem with connected pieces, it is useful as a building block for cake-cutting with disconnected pieces, as we describe next.

Given a pre-specified agent $i$ (which we call the "VIP"), we say that an allocation is \textbf{$\envyfreevip(n,M)$} if it is both envy-free and gives the VIP a piece worth at least $1/M$ of the total cake value. We prove the following reductions (Sections \ref{sec:disconnected-4}-\ref{sec:disconnected-n}):

\begin{lemma*}\label{lem:reductions}
(a) If we can find $\envyfreevip(n,M)$ allocations using $T$ queries then we can find an $\envyfree(n,M)$ allocation using $n\cdot T$ queries.

(b) If we can find $\envyfreevip(n,M)$ allocations using $T$ queries, then for every $\epsilon>0$ we can find an $\envyfreevip(n, \frac{n}{1-\epsilon})$ allocation using $\factor \cdot T$ queries.
\end{lemma*}

The allocations found by the algorithm of Theorem \ref{thm:4-agents} are not only $\envyfree (4,7)$ but also $\envyfreevip(4,4)$. This allows us to provide an alternative proof to a result of \citeN{Saberi2009Cutting}:\footnote{\citeN{Saberi2009Cutting} prove that their protocol is finite but do not calculate the number of queries.}
\begin{theorem*} \label{thm:4-agents-disc}
Assuming free disposal, it is possible to find 
an $\envyfree(4,4)$ allocation 
using at most 
	\begin{cuts}
	24 cuts and 
	\end{cuts}
260 queries.
\end{theorem*}
While the algorithm of Theorem \ref{thm:4-agents-disc} does not guarantee connected pieces, it uses only 24 cuts. This means that the cake is cut to at most 25 pieces, so each agent receives less than 7 pieces in average.

More importantly, the algorithm used in Theorem \ref{thm:n-agents} finds $\envyfreevip(n,2^{n-2}+1)$ allocations using $O(n\cdot 2^n)$ queries. This gives:

\begin{theorem*}\label{thm:n-agents-disc}
Assuming free disposal, for every $\epsilon>0$, it is possible to find an $\envyfree(n, \frac{n}{1-\epsilon})$ allocation using
	\begin{cuts}
	$O(4^n\ln(1/\epsilon))$ cuts and
	\end{cuts}
$O(n\cdot 4^n\ln(1/\epsilon))$ queries.
\end{theorem*}
This means that we can find, in bounded time, an envy-free division which is as close as we want to a proportional division (each agent receives at least $(1-\epsilon)/n$ of the total cake value). The number of queries is linear in the binary representation of the approximation factor ($\epsilon$).

All our algorithms are very simple and use only a single type of query: the Equalize query (defined in Section \ref{sec:tools}). The hard work is done in the correctness proofs. We view this simplicity as an additional advantage of the free disposal assumption. 

Table \ref{tab:summary} summarizes our results and compares them to some related work surveyed in the next subsection.

\newcommand{\prop}[1]{\textcolor{blue}{$1/#1$ *}}

\begin{table}
\captionsetup{singlelinecheck=off}
\small
\hskip-15mm
\begin{tabular}{|c|c|c|c|c|c|c|c|}
\hline Name & 
       Pieces & \small{Agents} & Valuations & 
       \#Cuts & \#Queries &  \small{Envy} & Prop. \\ 
\hline
\hline 
	\shortstack{
		\citet{Stromquist1980How} \\
		\citet{Robertson1998CakeCutting}\\
		\citet{Barbanel2004Cake}
	}
	& 
       Con. & 3 & General & 
       $2$ & Infinite & 0 & \prop{3}
\\ 
\hline
\citet{Su1999Rental} &
       Con. & $n$ & General & 
       $n-1$ & Infinite & 0 & \prop{n} 
\\ 
\hline 
\citet{Deng2012Algorithmic} &  
       Con. & $n$ & Lipschitz & 
       $n-1$ & $O((1/\epsilon)^{n-2})$ & $\epsilon$ & $1/n - \epsilon$
\\ 
\hline
\citet{Branzei2015Note} & 
       Con. &$n$ & Polynomials & 
       $n-1$ & $O(n^2\cdot d)$ & 0          & \prop{n}             \\ 

\hline Selfridge-Conway & 
       Dis. & 3 & General & 
       $4$ & 13 & 0 & \prop{3} \\ 
\hline \citet{Saberi2009Cutting} & 
       Dis. & 4 & General & 
       Const. & Const. & 0 & \prop{4} \\ 
\hline \shortstack{Reentrant-diminisher\\ \citep{Brams1996Fair}} &
       Dis. & $n$ & General & 
       ? & $O(n^2/\epsilon)$ & $\epsilon$ & \prop{n}   \\ 
\hline \shortstack{\citet{Brams1995EnvyFree}\\
       \citet{Robertson1998CakeCutting}\\\citet{Pikhurko2000EnvyFree}} &
       Dis. & $n$ & General & 
       Unbounded & Unbounded & 0 & \prop{n}   \\ 
\hline \citet{Kurokawa2013How} & 
       Dis. & $n$ & Piecewise-lin. & 
       ? & $O(n^6 k \ln{k})$ & 0 & \prop{n} \\ 
\hline
\hline \shortstack{\citet{Aziz2015Discrete}\\(see also our Appendix \ref{sec:entire-cake})} & 
    Dis. & 4 & General & 
    203 & 584 & 0 & \prop{4} \\ 
\hline \citet{Aziz2016Discrete} & 
       Dis. & $n$ & General & 
    ? & $n^{n^{n^{n^{n^{n}}}}}$ & 0 & \prop{n} \\ 
\hline \citet{Aziz2016Discrete} & 
      Dis. & $n$ & General & 
    ? &   $n^{n+1}$ & 0 & \prop{n} \\ 
\hline \citet{Aziz2016Discrete} & 
       Con. & $n$ & General & 
    $n-1$ & $n^{n+1}$ & 0 & $1/(3n)$ \\ 
\hline
\hline Section \ref{sec:n-agents} &
       Con. & $n$ & General & 
       $2^{n-1}-1$ & $O(n\cdot 2^n)$ & 0 & $1 / 2^{n-1}$ \\ 
\hline Section \ref{sec:3-agents} & 
       Con. & 3 & General & 
       3 & 54 & 0 & \prop{3} \\ 
\hline Section \ref{sec:4-agents} &
       Con. & 4 & General & 
       6 & 65 & 0 & $1/7$ \\ 
\hline Section \ref{sec:disconnected-4} &
       Dis. & 4 & General & 
       24 & 260 & 0 & \prop{4} \\ 
\hline Section \ref{sec:disconnected-n} &
       Dis. & $n$ & General & 
       $O(4^n\ln(1/\epsilon))$ & $O(n\cdot 4^n \ln{(1/\epsilon)})$ & 0 & \small{$1/n-\epsilon/n$} \\ 
\hline 
\end{tabular}
\protect\caption[comparison table]{
Envy-free cake-cutting algorithms.
Top section shows algorithms developed prior to this paper.
Middle section shows algorithms developed after the conference version of this paper.
Bottom section shows the algorithms presented in this paper. \label{tab:summary}
\\
\\
\textbf{Legend:}
\begin{itemize}
\item Pieces column: whether the pieces are \textbf{Con}nected or \textbf{Dis}connected.
\item Queries column: $d$ and $k$ are parameters of the valuation functions (maximum degree of polynomials and number of pieces, respectively).
\item Envy column: $\epsilon$ is an additive approximation constant (every agent values other pieces at most $\epsilon$ more than its own piece).
\item Prop column: the expression is the proportion of the total cake value that is guaranteed to all agents. \prop{n} implies that the division is proportional.
\end{itemize}
\vskip 3cm
} 
\end{table}

\subsection{Related work}

\subsubsection{Earlier cake-cutting algorithms}
Cake-divisions are commonly modeled in one of two ways: in the \emph{connected} model, the algorithm must give each agent a single contiguous piece; in the \emph{unrestricted} model, the algorithm may give each agent several disconnected pieces.

Proportional division is well understood from a computational perspective.  The algorithm of \citet{Steinhaus1948Problem} generates a proportional division with connected pieces in $O(n^2)$ queries, and an improved algorithm by \citet{Even1984Note} requires only $O(n \log{n})$ queries. Later results proved that this runtime is asymptotically optimal even if disconnected pieces are allowed \citep{Woeginger2007Complexity,Edmonds2011Cake}.

Envy-free division is a much harder task, even when only 3 agents are involved.  The first envy-free division algorithm for 3 agents with connected pieces was published by \citet{Stromquist1980How}. This algorithm is not discrete --- it requires the agents to simultaneously hold knives over the cake and move them in a continuous manner. This means that this algorithm cannot be accurately executed by a computer in finite time. There are simpler algorithms for the same task, e.g. \citet{Robertson1998CakeCutting}[pages 77-78] and \citet{Barbanel2004Cake}, but they also use moving-knives. A discrete and finite algorithm for envy-free division for 3 agents was constructed by Selfridge and Conway \citep{Brams1996Fair}. It requires only 13 queries\footnote{The first cutter makes 2 \emph{mark}s to create three equal pieces. The second makes 2 \emph{eval}s and 1 \emph{mark} to create two equal pieces. The third makes 2 \emph{eval}s to select the best piece. In the second phase, again the first cutter makes 2 \emph{mark}s, then the other two agents make 2 \emph{eval}s to select their best pieces.} but generates partitions with disconnected pieces.

Finding an envy-free division among four or more agents was a long-standing open problem. It was solved only in the 1990's, both for connected and disconnected pieces.  \citet{Su1999Rental} presented an algorithm, attributed to Forest Simmons, for envy-free division with connected pieces, but it is not finite --- it converges to an envy-free division after a possibly infinite number of queries. \citet{Brams1995EnvyFree}, \citet{Robertson1998CakeCutting} and \citet{Pikhurko2000EnvyFree} presented three different algorithms for envy-free division with disconnected pieces; while these algorithms are guaranteed to terminate in finite time, their run-time is not a bounded function of $n$.

Two important hardness results were proved in the 2000's. \citet{Stromquist2008Envyfree} proved that an envy-free division with connected pieces cannot be found by any finite algorithm, whether bounded or unbounded. This shows that the connectivity requirement makes the envy-free division problem strictly more difficult. Shortly afterward, \citet{Procaccia2009Thou} proved an $\Omega(n^2)$ lower bound on the query complexity of any envy-free division algorithm, even with disconnected pieces. This proved that the problem of envy-free division is strictly more difficult than the problem of proportional division.

\subsubsection{Later cake-cutting algorithms}
After the publication of the conference version of this paper \citep{SegalHalevi2015Waste}, several groundbreaking results have been published by Haris Aziz and Simon Mackenzie. Initially \citep{Aziz2015Discrete} they published the first bounded-time algorithm for envy-free division of an entire cake among 4 agents. Their algorithm uses at most 584 queries and 203 cuts, which means that the cake is cut to 204 pieces (so each agent receives 51 pieces in average). In Appendix \ref{sec:entire-cake} we provide a somewhat simpler presentation of this result, based on our Algorithm \ref{alg:4-agents}.

Later, \citet{Aziz2016Discrete} published the first bounded-time algorithm for envy-free division of an entire cake among $n$ agents. Their algorithm uses $n^{n^{n^{n^{n^{n}}}}}$ queries, which is also an upper bound on the number of cuts. The core subroutine in that paper finds an $\envyfreevip(n,n)$ allocation of a part of the cake using $n^n$ queries. By Lemma \ref{lem:reductions}, this implies that an $\envyfree(n,n)$ allocation with disconnected pieces can be found using $n^{n+1}$ queries assuming free disposal. This result provides an affirmative answer to an open question we posed at the end of \citet{SegalHalevi2015Waste}. In their latest working paper\footnote{http://arxiv.org/abs/1604.03655v7}
they present an algorithm that finds an $\envyfree(n,3n)$ allocation with connected pieces using $O(n^{n+1})$ queries, assuming free disposal. This is a great improvement over our Algorithm \ref{alg:n-agents}. Finding an $\envyfree(n,n)$ allocation with connected pieces in finite time is still an open problem.

Recently, \citet{Amantidis2018Improved} presented an algorithm for envy-free cake-cutting for four agents, which improves over the one by Aziz and Mackenzie by being simpler and requiring less cuts and queries. 

\subsubsection{Approximations}
Cake-cutters have tried to cope with the difficulty of envy-free division in several ways.

One way is to relax the envy-freeness criterion and allow a small amount of envy. \citet{Brams1996Fair}(pages 130-131) describe a re-entrant variant of Steinhaus'  algorithm which produces a division with disconnected pieces in which the envy of every agent is at most an additive constant $\epsilon$ (for every agent, the value of its piece plus $\epsilon$ is at least the value of any other piece). The number of queries is polynomial in $n$ and linear in $(1/\epsilon)$. \citet{Deng2012Algorithmic} present a similar approximation with connected pieces; here the number of queries is exponential in $n$ and polynomial in $(1/\epsilon)$.  In contrast to these results, our algorithms guarantee full envy-freeness. Our algorithm for disconnected piece also guarantees an additive approximation to proportionality. The number of queries in our approximation is exponential in $n$ but \emph{logarithmic} in $(1/\epsilon)$ (in other words, it is linear in the binary representation of the approximation constant).

A second way is to restrict the value functions of the agents. \citet{Kurokawa2013How} require the value functions to be piecewise-linear and find an envy-free division with disconnected pieces in time polynomial in the size of the representation of the value functions. \citet{Deng2012Algorithmic} require the value functions to be Lipschitz-continuous and find an approximately-envy-free division with connected pieces. \citet{Branzei2015Note} requires the value functions to be polynomials of bounded degree and finds an envy-free division with connected pieces in time polynomial in the maximum degree. In contrast to these results, our algorithms apply to arbitrary non-atomic value functions, and their runtime guarantee is a function of only the number of agents but not the peculiarities of their valuation functions.

\subsubsection{Free disposal}
The free disposal assumption was introduced into envy-free cake-cutting by \citet{Saberi2009Cutting}. They used it only for 4 agents and disconnected pieces. 

Later, free disposal has also been studied by \citet{Arzi2011Throw}. They proved that discarding some parts of the cake may allow us to achieve an envy-free division with an improved social welfare (i.e. the sum of the utilities of the agents is larger than in the no-free-disposal case). They call this phenomenon the \emph{dumping paradox}. Our paper demonstrate a different kind of a dumping paradox --- we show that dumping some parts of the cake can be beneficial not only from an economic perspective but also from a computational perspective. 

A third scenario in which free disposal is required is when the pieces must have a pre-specified geometric shape, such as a square \citep{SegalHalevi2015EnvyFree}. 

There is some related work concerning allocation of indivisible goods where the same idea of not allocating all the objects is used to get better fairness results \citep{Brams2013TwoPerson,Aziz2015Generalization}.\footnote{We thank an anonymous referee for referring us to these papers.}

Partial-proportionality was introduced by \citet{Edmonds2006Balanced} and \citet{Edmonds2008Confidently}. They used it, like us, to reduce the query complexity. Their algorithm gives each agent at least $1/(a\cdot n)$ of the total value, where $a\geq 10$ is some sufficiently large constant, with a query complexity of $O(n)$. This is better than the optimum of $O(n \log{n})$ required for finding a fully-proportional division. Their algorithm is not envy-free.

\subsubsection{Computational models}
The most prominent computational model for discrete cake-cutting is the \emph{mark-eval model} of \citet{Robertson1998CakeCutting}: an \emph{eval query} asks an agent to reveal its value for a specified piece of cake; a \emph{mark query} asks an agent to mark a piece of cake with a specified value. Our algorithms use a single primitive query --- Equalize (defined in Section \ref{sec:tools}). We prove in Lemma \ref{lem:equalize} that Equalize can be implemented by a bounded number of mark-eval queries, so all our algorithms are bounded in the standard model.

A different model, the \emph{cut-choose model}, was recently suggested by \citet{Branzei2016Algorithmic}. In this model, a \emph{cut query} asks an agent to cut one of a subset of the cake-pieces currently on the table; a \emph{choose query} asks an agent to select one of a subset of pieces. Our Equalize query is just a bounded sequence of \emph{cut}s, so our algorithms our bounded in their model, too.

\subsection{Paper structure}
The model is formally defined in Section \ref{sec:The-Model}. The main tools used in our division algorithms, the \emph{preference graph} and the \emph{Equalize query}, are introduced in Section \ref{sec:tools}. 

Sections 4-8 are devoted to proving the theorems. For every $k\in\{1,\ldots,5\}$, Theorem $k$ is proved in Algorithm $k$ and Section $k+3$.

A detailed, computer-generated correctness proof of the four-agents algorithm of Section \ref{sec:4-agents} is given in Appendix \ref{sec:proof-4-agents}. An application of that algorithm to an envy-free division of an entire cake (as in \citet{Aziz2015Discrete}) is presented in Appendix \ref{sec:entire-cake}.

\section{Model and Notation}

\label{sec:The-Model}
The cake is assumed to be the unit interval $[0,1]$.

There are $n$ agents, denoted by \agent{1},\agent{2},...,\agent{n}. When the number of agents is small, they are denoted instead by A,B,C,... or by Alice,Bob,Carl...

An \emph{allocation} of a cake is an $n$-tuple of pairwise-disjoint subsets of the cake: $X_1\cup\dots\cup X_n \subseteq [0,1]$. When \emph{connected} pieces are required, each piece $X_i$ must be an interval; when \emph{disconnected} pieces are allowed, each piece $X_i$ may be a finite union of intervals.

Each agent \agent{i} has a preference relation $\succeq_i$ that is represented by a non-negative value measure $V_i$ on the pieces. The term "measure" implies that it is additive --- the value of a piece is equal to the sum of the values of its parts. All value measures are absolutely continuous with respect to length. This implies that all singular points have a value of 0 to all agents, i.e. there are no valuable ''atoms'' which cannot be divided. The value measures are normalized such that $\forall i: V_i([0,1])=1$. All these assumptions are standard in the cake-cutting literature.

An allocation $X$ is called \emph{envy-free} if each agent values his allocated piece at least as much as every other allocated piece:

\[
\forall i,j\in\{1,...,n\}:\,\, V_i(X_i) \geq V_i(X_j)\,\,\,\,\,\,(\textrm{Equivalently:}\,\,\,\,\,\, X_i \succeq_i X_j)
\]
We say that an allocation $X$ has a \emph{proportionality} of $1/M$ if it allocates each agent a fraction of at least $1/M$ of the total cake value:
\[
\forall i\in\{1,...,n\}:\,\, V_i(X_{i}) \geq 1/M
\]
An allocation with a proportionality of $1/n$ is usually called a \emph{proportional} allocation. 

An allocation $X$ is called \textbf{$\envyfree(n,M)$} if it is both envy-free and has a proportionality of $1/M$. Note that every envy-free allocation of the entire cake is \envyfree (n,n),\footnote{An envy-free allocation gives each agent a piece which is best (for that agent) of $n$ pieces. By the pigeonhole principle, the best of $n$ is worth at least $1/n$. Hence, an envy-free allocation of an entire cake has a proportionality of $1/n$.}  but this is not necessarily true when some cake remains unallocated.

Given a pre-specified agent $\agent{j}$,  an allocation $X$ is called \textbf{$\envyfreevip(n,M)$} with VIP $\agent{j}$ if it is envy-free, and additionally the VIP receives a value of at least $1/M$:
\begin{align*}
V_j(X_{j}) \geq 1/M
\end{align*}

\section{Tools} \label{sec:tools}
Our algorithms are described in a bottom-up approach. We first present basic tools that perform well-defined tasks, then combine these tools to get a full algorithm. We believe that the bottom-up approach may be beneficial to future cake-cutters, that may use our tools to develop improved algorithms.

\subsection{The preference graph}
At any time during the execution of an algorithm, there is a certain number of pieces on the table, which together comprise the entire cake. The \emph{preference graph} is a bipartite graph, in which the nodes in one side represent the $n$ agents and the nodes in the other side represent the pieces. The pieces are denoted as numbers with a hat, e.g. \piece{1}, \piece{2}, etc.  There is an edge from an agent \agent{k} to a piece \piece{i} if \agent{k} \emph{prefers} \piece{i}, i.e., for every piece \piece{j}: $\piece{i} \succeq_k \piece{j}$. Note that an agent can "prefer" two or more pieces. This means that the agent is indifferent between these pieces but values any of them more than any other piece. Here are two possible preference graphs for three agents:\\

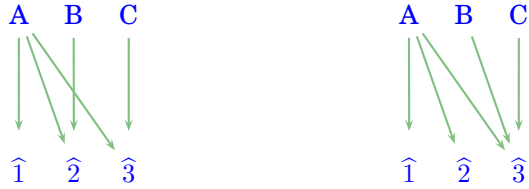
\begin{center}
\begin{psmatrix}[colsep=3cm,rowsep=0mm]
	\begin{psmatrix}[colsep=0cm,rowsep=0mm]
		\threeagents{}
	
		\like{A}{1}\like{A}{2}\like{A}{3}
		\like{B}{2}
		\like{C}{3}
	\end{psmatrix}
& 
	\begin{psmatrix}[colsep=0cm,rowsep=0mm]
		\threeagents{}
	
		\like{A}{1}\like{A}{2}\like{A}{3}
		\like{B}{3}
		\like{C}{3}
	\end{psmatrix}
\end{psmatrix}
\end{center}

Both graphs may be the result of Alice cutting the cake to 3 pieces which are equal in her eyes. In the left graph, Bob and Carl each prefer a different piece; in the right graph, they prefer the same piece (\piece{3}).

A \emph{saturated matching} in a bipartite graph is a subset of the edges, in which each agent-node has a single neighbor and each piece-node has at most a single neighbor. A saturating matching in the preference graph corresponds to an envy-free allocation of a part of the cake, since every agent is allocated a preferred piece.

A well-known tool for proving the existence of saturated matchings in bipartite graphs is \emph{Hall's marriage theorem}. This theorem, applied to our setting, implies the following lemma:

\begin{lemma}
\label{lem:unknown0x}
If for every $k\range{1}{n}$, every group of $k$ agents jointly prefers at least $k$ pieces, 
then an envy-free allocation exists. 
\end{lemma}

In a preference graph, Hall's condition is always satisfied for groups of $k=1$ agents since every agent has at least one preferred piece. 

In the left graph above, Hall's condition is also satisfied for every group of 2 or 3 agents; this means that an envy-free allocation exists. Indeed, the allocation A-\piece{1}, B-\piece{2} and C-\piece{3} is envy-free.

In the right graph above, Hall's condition is violated by the group \{B,C\}. In this case, to get an envy-free allocation, the graph should be \emph{transformed} in order to create a graph that meets Hall's condition. The main query we use to transform the preference graph is the \emph{Equalize} query, which is described in the next subsection. 

\subsection{The Equalize query}  \label{sub:equalize}
Given an integer $k\geq 2$, the query \equalize(k) asks an agent to mark zero or more pieces such that, if the pieces are cut according to these marks, that agent will have at least $k$ best pieces. For example, in the right graph above, an \equalize(2) query to Bob implies the following question: "where would you cut piece \piece{3}, your currently favorite piece, such that you will have two equally-best pieces?". 

Suppose Bob's second-best piece is \piece{2}. Bob can answer the \equalize(2) question in one of two ways:

\begin{enumerate}
\item If $V_B(\piece{3}) \geq 2 V_B(\piece{2})$, then $\piece{3}$ should be cut to two pieces of equal value, which is $V_B(\piece{3})/2$.
\item Otherwise, $\piece{3}$ should be cut to two unequal pieces --- one having a value of $V_B(\piece{2})$ and the other having a smaller value $V_B(\piece{3})-V_B(\piece{2})$. 
\end{enumerate}

Note that an \equalize(2) query can be implemented by a constant number of \emph{mark} and \emph{eval} queries of the standard model \citep{Robertson1998CakeCutting}. In Section \ref{sec:n-agents} below we prove that the same is true for $k>2$; see Lemma \ref{lem:equalize}.

A third option is that $V_B(\piece{3})=V_B(\piece{2})$. In this case, no cutting is needed since Bob already has two pieces of equal value and better than the third piece. Here and in the rest of the paper, we ignore such fortunate coincidences. This does not lose generality, since it only makes it harder to find an envy-free division --- it decreases the number of edges in the preference graph and makes it harder to find a saturated matching. 

Formally, we make the following assumption about the preference graph:
\begin{assumption}\label{asm:diff}
After an agent $A$ cuts a piece \piece{i}, if an agent $B\neq A$ prefers a piece $\piece{j}\neq\piece{i}$, then $B$ does not prefer \piece{i}.
\end{assumption}
Assumption \ref{asm:diff} makes the descriptions of division algorithms simpler since it reduces the number of cases that need to be handled. We now prove that this simplicity does not lose generality.
\begin{lemma}\label{lem:asm:diff}
If there is an algorithm $P$ that finds \envyfree(n,M) allocations when Assumption \ref{asm:diff} is satisfied, then there is an algorithm $P'$ that finds \envyfree(n,M) allocations even when Assumption \ref{asm:diff} is violated.
\end{lemma}
\begin{proof}
We assume that the existing algorithm $P$ is given as a service, that accepts the current preference graph and replies with the next Equalize query to issue.

The new algorithm $P'$ uses this service to simulate $P$. It issues the Equalize query sent by $P$ to the agents, collects the agents' replies, makes the required cuts and updates the preference graph. If the new preference graph violates assumption \ref{asm:diff}, so that e.g. after agent $A$ cut \piece{i} agent $B$ prefers both $\piece{j}\neq\piece{i}$ and $\piece{i}$, then $P'$ removes the edge $B\to \piece{i}$ and sends to $P$ the reduced graph. The modified graph is also a possible outcome of \equalize(k) and it satisfies assumption \ref{asm:diff}, so $P$ must know how to handle it. Hence, eventually the simulation is terminated and the reduced preference graph has a saturated matching which corresponds to an \envyfree(n,M) allocation. The same matching is also a saturated matching on the real preference graph, since the real preference graph contains (at least) all the edges of the reduced graph. Hence, $P'$ returns an \envyfree(n,M) allocation.
\end{proof}

Assumption \ref{asm:diff} has several simple corollaries which are implicitly used below. For every set $X$ of pieces, define the \emph{last cutter} of $X$ to be the last agent who made a cut on any piece from the set $X$.
\begin{itemize}
\item If an agent $A$ prefers a set $X$ of pieces with $|X|\geq 2$, then $A$ is the last cutter of $X$.
\item Each two agent-nodes in the preference graph have at most one neighbor at common (there is at most one piece that both agents prefer).
\item If an agent cuts the cake to several equal pieces, then every other agent prefers exactly one piece (as in the graphs above).
\end{itemize}

We now return to the three-agent example. If the algorithm implements Bob's suggested cuts, the preference graph is transformed. If Bob gave an answer of type (1), it is transformed as in the left graph below; if Bob gave an answer of type (2), it is transformed as in the right graph:

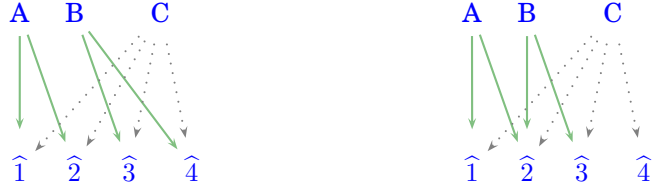
\begin{center}
\begin{psmatrix}[colsep=3cm,rowsep=0mm]
	\begin{psmatrix}[colsep=0cm,rowsep=0mm]
		\threeagents{}\piecenode{4}
	
		\like{A}{1}\like{A}{2}
		\like{B}{3}\like{B}{4}
		\likemaybe{C}{1}\likemaybe{C}{2}\likemaybe{C}{3}\likemaybe{C}{4}
	\end{psmatrix}
& 
	\begin{psmatrix}[colsep=0cm,rowsep=0mm]
		\threeagents{}\piecenode{4}
	
		\like{A}{1}\like{A}{2}
		\like{B}{3}\like{B}{2}
		\likemaybe{C}{1}\likemaybe{C}{2}\likemaybe{C}{3}\likemaybe{C}{4}
	\end{psmatrix}
\end{psmatrix}
\end{center}

Note that in both cases, the edge A-\piece{3} is gone, because piece \piece{3} has been trimmed by Bob so its value for Alice is probably smaller. By Assumption \ref{asm:diff}, we ignore the fortunate coincidence in which Bob trimmed a part which happens to be worthless for Alice. The edges A-\piece{1} and A-\piece{2} remain, because these two pieces were not touched by Bob (Alice is still the last cutter of these pieces).

The dotted edges emanating from C imply that we do not know which piece is preferred by Carl after the cuts, since his previously-best piece --- \piece{3} --- has been trimmed. On the other hand, we know that Bob now prefers two pieces --- one of them is the trimmed \piece{3} and the other is another piece, which is either his previously-second-best piece \piece{2} or the new piece \piece{4}.

Even though we don't know which piece is now preferred by Carl, we can be sure that a saturated matching exists. This follows from the following lemma:

\begin{lemma}
\label{lem:unknown1x}
Suppose the $n$ agents are divided to $n-1$ agents whose preferences are known and one agent whose preferences are unknown. If for every $k\range{1}{n-1}$, every group of $k$ known agents jointly prefers at least $k+1$ pieces, then an envy-free allocation exists.
\end{lemma}
\begin{proof}
Suppose every group of $k$ known agents jointly prefers at least $k+1$ pieces. Then, for any possible preference of the unknown agent, every group of $k+1$ agents jointly prefers at least $k+1$ pieces. Additionally, every group of 1 agent always prefers at least 1 piece. Hence, by Lemma \ref{lem:unknown0x} an envy-free allocation exists.
\end{proof}

In the graphs above, there are two known agents --- A and B, and one unknown agent --- C. Each of the known agents prefers two pieces, and the two of them together prefer 3 or 4 pieces. Hence, whatever C's preferences are, a saturated matching exists and an envy-free allocation can be found.

\subsection{Example: an algorithm for 3 agents} \label{sub:equalize-example}
By now, we have described an envy-free division algorithm for three agents. The algorithm can be succinctly summarized by the following two statements:

\begin{framed}
Alice: \equalize(3)

Bob: \equalize(2)
\end{framed}

In words: the algorithm asks Alice to cut the cake to 3 equal pieces in her eyes, then asks Bob to cut one of these pieces in order to make 2 equally-best pieces in his eyes. The outcome always looks like one of the preference graphs above, which means that a saturated matching exists and each agent can be allocated a best piece.

The last step of the algorithm is to actually find the matching and implement the corresponding envy-free allocation. To do this in our case, it is sufficient to ask Carl to pick his best piece, then ask Bob to pick one of his best pieces (which must be the piece that he trimmed, if it is still available), then ask Alice to pick a remaining piece. In the rest of this paper, we suppress this last step from the description of our algorithms. Since a maximum matching in a bipartite graph can always be found in polynomial time, it is sufficient to prove that the algorithm guarantees that a saturated matching exists.

\subsection{The Envy-Free-Proportionality Lemma}
We now calculate the proportionality of the resulting allocation --- the value guarantee per agent. This is based on a general lemma which we call the EFP (Envy-Free-Proportionality) lemma:

\begin{lemma} (EFP Lemma)
If a cake is partitioned to a set of $M \geq n$ pieces and each agent receives a single preferred piece from that set, then the allocation is \envyfree(n,M).
\end{lemma}
\begin{proof}
Envy-freeness is obvious since each agent receives one of his best pieces. Proportionality is a result of the fact that the value functions of the agents are measures, so they are additive. The sum of the values of all pieces is the value of the entire cake. Hence, by the pigeonhole principle, the value of any best piece is at least $1/M$ of the total cake value.
\end{proof}

Going back to our three-agents algorithm, we see that the algorithm partitions the cake to $M=4$ pieces. Hence, by the EFP lemma, it generates an \envyfree(3,4) allocation. This is only a warm-up algorithm; the algorithm of Section \ref{sec:3-agents} generates an \envyfree(3,3) allocation, which are optimal (in terms of proportionality) for 3 agents.

Before continuing with more advanced division algorithms, we need two more tools: an assumption and a lemma.

\subsection{The new-pieces assumption}

In all algorithms presented in this paper, the first query is an Equalize query, asking one of the agents to cut $ֵN$ equal pieces, where $N\geq n$ is some constant (like the "Alice:Equalize(3)" in Subsection \ref{sub:equalize-example}). The following queries are \equalize(k) queries where $k\leq N$. In the rest of this paper, we make the following additional assumption on the preference graph starting with the second query:

\begin{assumption}\label{asm:newpiece}
In any query after the first query, when a new piece is created by a cut, it is not the preferred piece of any agent.
\end{assumption}
So in the example of Subsection \ref{sub:equalize}, we assume that after "Bob:\equalize(2)" the preference graph looks like the rightmost graph. This assumption does not lose generality because, from Hall's perspective, it only makes it harder to find a saturated matching --- it concentrates the same number of edges over a smaller number of piece nodes. Formally:

\begin{lemma}\label{lem:asm:newpiece}
If there is an algorithm $P$ that finds \envyfree(n,M) allocations when Assumption \ref{asm:newpiece} is satisfied, then there is an algorithm $P'$ that finds \envyfree(n,M) allocations even when Assumption \ref{asm:newpiece} is violated.
\end{lemma}
\begin{proof}
Similarly to Lemma \ref{lem:asm:diff}, $P'$ simulates $P$ by reading the next Equalize query, issuing it to the agents and updating the preference graph.

Suppose the $m$-th query in $P$ ($m\geq 2$) is \equalize(k). This query creates $k-1$ new pieces. The algorithm $P'$ defines an injection $f_m$, which maps every new piece $\piece{j}$ to a unique original piece $f_m(\piece{j})$. By "original piece" we mean one of the pieces generated by the first Equalize query. There are at least $N$ original pieces and $k\leq N$, so an injection $f_m$ always exists. 

$P'$ then constructs a reduced preference graph, by converting any edge $A_i \to \piece{j}$ (for every agent $A_i$ and new piece \piece{j}) to an edge $A_i \to f_m(\piece{j})$. $P'$ then sends the reduced graph to $P$. The reduced graph corresponds to a possible outcome of \equalize(k) and it satisfies Assumption \ref{asm:newpiece}, so $P$ must know how to handle it. Hence, eventually the simulation is terminated and the reduced preference graph has a saturated matching. This matching corresponds to an \envyfree(n,M) allocation in which each \agent{i} receives a piece $X_i$, which is one of the original pieces.

For every original piece $X_i$ where $i\range{1}{N}$, define the set $Y_i = \{X_i\}\cup\{\piece{j}|f_m(\piece{j})=X_i, m\geq 2\}$. This is the set of all pieces that were mapped to $X_i$ at some point during the simulation of $P$. The sets $Y_i$ are pairwise-disjoint, since each new piece is mapped to a single original piece.

In the reduced graph, there is a preference edge $\agent{i} \to X_i$. By construction, this preference edge comes from some edge in the real graph, $\agent{i} \to \piece{j}$ where $\piece{j}\in Y_i$. Hence, for every \agent{i}, the set $Y_i$ contains some piece which is a best piece for \agent{i}.

As a final step, $P'$ asks each \agent{i} to select a best piece from the set $Y_i$. The total number of pieces ($M$) is unchanged. Hence, the resulting allocation is \envyfree(n,M).
\end{proof}

Hence, the new pieces are omitted from the preference graphs; only their total number is kept in mind for the sake of calculating the proportionality.

\subsection{The unknown-agent lemma}
\begin{lemma}
\label{lem:unknown1}
Suppose the $n$ agents are divided to $n-1$ agents whose preferences are known and one agent whose preferences are unknown. If every known agent prefers at least 2 pieces, then an envy-free allocation exists.
\end{lemma}
\begin{proof}
By Lemma \ref{lem:unknown1x}, it is sufficient to prove that every  $k$ known agents jointly prefer at least $k+1$ pieces. The proof is by induction on $k$. The base $k=1$ is given. Suppose the lemma is true for all groups of less than $k$ known agents. Consider a group of $k$ known agents $\agent{1},\dots,\agent{k}$. Each of these agents prefers at least two pieces. By Assumption \ref{asm:diff}, each of these agents was the last one to cut at least one of his two preferred pieces. Suppose that the last agent to cut one of his preferred pieces was \agent{1}. Suppose that \agent{1} prefers pieces \piece{1},\piece{2} and that he was the last agent to cut \piece{1}. 

By Assumption \ref{asm:diff} again, any other agent that prefers \piece{1} does not prefer any other piece. This means that any other known agent does not prefer \piece{1}. By the induction assumption, agents $\agent{2},\dots,\agent{k}$ jointly prefer $k$ pieces, which must be different than \piece{1}. With \piece{1}, agents $\agent{1},\dots,\agent{k}$ jointly prefer $k+1$ pieces.
\end{proof}

\section{Connected pieces and \MakeLowercase{\LARGE {$n$}} Agents}  \label{sec:n-agents}
Generalizing the 3-agent algorithm from the Section \ref{sec:tools} to $n$ agents requires the following building blocks: an \equalize(k) query for arbitrary $k$, and a generalized version of Lemma \ref{lem:unknown1}. We now describe each of these generalizations in turn.

Answering an \equalize(k) query for $k\geq 3$ is a non-trivial task. There are many different options. For example, a reply to \equalize(3) can have one of the following forms:
\begin{enumerate}
\item Cutting the best piece to three equal pieces, which are all better than the previously second-best piece; or -
\item Trimming the best piece such that it is twice as valuable as the second-best piece, then cutting the result to two halves; or -
\item Trimming both the best and the second-best pieces, such that the trimmed pieces are equal to the third-best piece.
\end{enumerate}
Naturally, the number of options grows as $k$ becomes larger.

Fortunately, \equalize(k) can be answered using a bounded number of \emph{mark} and \emph{eval} queries. We prove this by reducing \equalize(k) to the following problem:
\begin{quote}
\textbf{EnvyFreeStickDivision[$m$,$k$]}:
Given $m$ sticks of different lengths, make a minimal number of cuts such that there are at least $k$ pieces with equal lengths and no other piece is longer. 
\end{quote}
\citet{Reitzig2015Efficient} have recently presented an algorithm that solves the envy-free stick-division problem in time $O(m)$. The algorithm works as follows. First, based on the lengths of the $m$ sticks, it calculates an "optimal length", $l^*$. This is defined as the largest $l$ such that it is possible to cut at least $k$ pieces of length $l$. Then, it cuts $l^*$-sized pieces off of any stick longer than $l^*$ until all sticks
have length at most $l^*$. We use their algorithm to prove the following lemma.
\begin{lemma}
\label{lem:equalize}
When there are $m$ pieces on the table, an agent's answer to an \equalize(k) query can be calculated using $m-1$ \emph{eval} queries and $k-1$ \emph{mark} queries.
\end{lemma}
\begin{proof}
Denote the pieces currently on the table by $X_1,\dots,X_m$. Use $m-1$ \emph{eval} queries to find the agent's valuations to these pieces, $V_i(X_1),\dots,V_i(X_{m})$.\footnote{only $m-1$ \emph{eval}s are needed, since the $m$-th value can be calculated based on the other values and the additivity of $V_i$.} Create $m$ sticks, such that the length of stick $j$ is $V_i(X_j)$. Use the algorithm of \citeN{Reitzig2015Efficient} for EnvyFreeStickDivision[$m$,$k$] to find the optimal length $l^*$. Using $k-1$ \emph{mark} queries, mark $k$ pieces that the agent values as exactly $l^*$.\footnote{Only $k-1$ \emph{mark}s are needed, since for the maximum length $l^*$, at least one stick is cut evenly with no remainder.}
By definition of EnvyFreeStickDivision, the values of all other pieces are at most $l^*$, so this is a correct answer to \equalize(k).
\end{proof}

The next tool we need is a generalization of Lemma \ref{lem:unknown1}.

\begin{lemma}
\label{lem:unknownn}
Suppose the $n$ agents are divided to $n-u$ agents whose preferences are known and $u$ agents whose preferences are unknown. If every known agent prefers at least $1+2^{u-1}$ pieces, then an envy-free allocation can be attained with a bounded number of queries.
\end{lemma}
\begin{proof}
The proof is by induction on $u$. The base $u=1$ is Lemma \ref{lem:unknown1}. Assume the claim is true for $u$; we have to prove it for $u+1$ unknown agents. 

Suppose the known agents are $\agent{1},\dots,\agent{n-u-1}$ and the unknown agents are $\agent{n-u},\dots,\agent{n}$. Suppose that every known agent prefers at least $1+2^{u}$ pieces. We have to prove that an envy-free allocation can be attained with a bounded number of queries.

Ask agent \agent{n-u} to $Equalize(1+2^{u-1})$. This requires it to trim at most $2^{u-1}$ pieces and guarantees that it prefers $1+2^{u-1}$ pieces. Every known agent still prefers at least $(1+2^{u})-2^{u-1}\, =\, 1+2^{u-1}$ pieces. Hence, by adding \agent{n-u} to the set of known agents, the situation becomes exactly as in the induction assumption: there are $u$ unknown agents each of whom prefers $1+2^{u-1}$ pieces. Hence, by the induction assumption an envy-free allocation can be attained with a bounded number of queries.
\end{proof}

The proof of Lemma \ref{lem:unknownn} above immediately translates to a division algorithm (Algorithm \ref{alg:n-agents}). Initially, all agents are unknown. \agent{1} is asked to $Equalize(2^{n-2}+1)$ and becomes a known agent. Each step, another agent is asked to Equalize and becomes a known agent. Finally, \agent{n-1} is asked to \equalize(2); then, each of the first $n-1$ agents prefers at least 2 pieces, and by Lemma \ref{lem:unknown1} a saturated matching exists.\footnote{The matching can be implemented by letting the agents pick pieces in reverse order, from $n$ to $1$.}

\begin{algorithm}
\caption{Finding \envyfree (n,2^{n-1}) allocations with connected pieces.}
\label{alg:n-agents}

For $u$ = $n-1$ to $1$:

\hspace{1cm} \agent{n-u}: \equalize(2^{u-1}+1).
\end{algorithm}

Each Equalize action requires $2^{u-1}$ cuts. Hence the total number of cuts required is:
\begin{align*}
\sum_{u=1}^{n-1}{2^{u-1}} = 2^{n-1}-1
\end{align*}
and the total number of pieces after the last cut is $2^{n-1}$. By the EFP Lemma, the division is envy-free and each agent receives a connected pieces with value at least: 
\begin{align*}
\frac{1}{2^{n-1}}
\end{align*}
The algorithm uses $n$ Equalize queries, each of which can be calculated using $O(2^n)$ mark-eval queries. This proves our Theorem \ref{thm:n-agents}.

\begin{remark}
The algorithm presented above is similar to an algorithm mentioned by \citet{Brams1996Fair} (chapter 7, page 135)
as a sub-routine of their unbounded algorithm for envy-free cake-cutting with disconnected pieces. However, their sub-routine does not use the generalized Equalize query and hence does not guarantee any positive proportionality with connected pieces.
\end{remark}

\section{Connected pieces and 3 Agents}  \label{sec:3-agents}
Our goal in this and the next section is to improve the proportionality from the exponential figure guaranteed by the algorithm of Section \ref{sec:n-agents}. In this section we focus on the case of 3 agents. We first note that any algorithm starting with a pre-specified agent cutting 3 equal pieces cannot guarantee a proportionality of more than 1/4 with connected pieces (since the values of these pieces in the eyes of the other two agents might be 1/2, 1/4 and 1/4). However, when the cutting agent can be selected according to preferences, the optimal proportionality --- 1/3 --- is attainable. This can be done by Algorithm \ref{alg:3-agents}.

\begin{algorithm}
\caption{Finding \envyfree (3,3) allocations --- envy-free and proportional for 3 agents with connected pieces.}
\label{alg:3-agents}

\dirtree{
.1 One of:.
.2 Alice:\equalize(3) .
.2 Alice:\equalize(3); Bob:\equalize(2) .
.2 Alice:\equalize(3); Carl:\equalize(2) .
.2 Bob:\equalize(3) .
.2 Bob:\equalize(3);   Alice:\equalize(2) .
.2 Bob:\equalize(3);   Carl:\equalize(2) .
.2 Carl:\equalize(3) .
.2 Carl:\equalize(3); Alice:\equalize(2) .
.2 Carl:\equalize(3); Bob:\equalize(2) .
}
\end{algorithm}

The "One of" statement means that the algorithm should try each of the 9 execution branches on paper, and check whether it "succeeds" (i.e, leads to an envy-free and proportional division). Whenever an execution branch succeeds, the algorithm stops and implements the resulting allocation on the real cake. We now prove that at least one branches indeed succeeds.

\begin{lemma}
For every preferences of the agents, there is at least one branch of Algorithm \ref{alg:3-agents} in which the resulting allocation is $\envyfree (3,3)$.
\end{lemma}
\begin{proof}
It is convenient to normalize the valuations such that each agent values the entire cake as 3. Hence, proportionality requires that each agent receives a value of at least 1. 

Note that in each branch, the agent that does the initial \equalize(3) always has at least one whole piece to choose, so he always feels no envy and has a value of at least 1. It remains to prove that the same is true for the other two agents. I.e, in at least one branch, there is an envy-free allocation and for every agent there is a piece worth at least 1.

Assume, for the sake of the proof, that each agent marks two points in the interval $[0,1]$ that partition it to three intervals equal in his eyes. Denote the equal pieces of agent X by: \agentspiece{X}{1}, \agentspiece{X}{2} and \agentspiece{X}{3}, such that the value of \agentspiece{X}{i} to agent X is exactly 1.

Assume w.l.o.g. that the order of the first lines is A-B-C. Hence: $\agentspiece{A}{1}\subseteq \agentspiece{B}{1}\subseteq \agentspiece{C}{1}$.  There are $3!=6$ options for the order of the second lines.\footnote{We ignore the fortunate case in which two or more agents make a mark in the exact same spot. This case can be handled by assuming an arbitrary order between these agents.} We treat each of these cases in turn. Each case is illustrated by a picture; the vertical gray lines in the pictures are the cuts made by the agent who does the "\equalize(3)" in the successful branch.

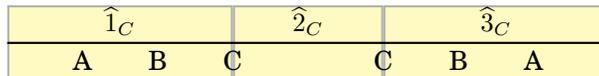
\psset{unit=1mm}

\subsection{C-B-A}

\begin{center}
\begin{pspicture}(80,15)
\agentspieceframe{C}{1}{0}{30}
\agentspieceframe{C}{2}{30}{50}
\agentspieceframe{C}{3}{50}{80}
\agentcuts{C}{B}{A}
\end{pspicture}
\end{center}
Ask Carl to \equalize(3) by cutting the cake at the points marked by "C" in the above picture (the vertical gray lines). Both Alice and Bob value two pieces ---  \agentspiece{C}{1} and \agentspiece{C}{3} --- as at least 1. This can be easily seen in the picture. E.g, 
the fact that Alice's leftmost mark is inside \agentspiece{C}{1} means that Alice values a subset of  \agentspiece{C}{1} as exactly 1, so she values \agentspiece{C}{1} as at least 1. The same is true for piece \agentspiece{C}{3} and for Bob.

Ask either Alice or Bob to \equalize(2). At most one piece is trimmed, so for each agent, at least one remaining piece has value at least 1. Moreover, both the cutter and Carl have two preferred pieces, so by Lemma \ref{lem:unknown1} an envy-free allocation exists. So the branch "Carl:\equalize(3); Alice:\equalize(2)" succeeds (the branch "Carl:\equalize(3); Bob:\equalize(2)" also succeeds, but we only need one successful branch).

\subsection{C-A-B}
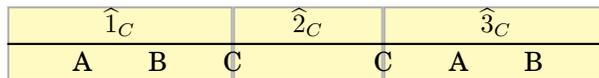
\begin{center}
\begin{pspicture}(80,15)
\agentspieceframe{C}{1}{0}{30}
\agentspieceframe{C}{2}{30}{50}
\agentspieceframe{C}{3}{50}{80}
\agentcuts{C}{A}{B}
\end{pspicture}
\end{center}
The analysis of the case C-B-A applies as-is to this case.

\subsection{A-B-C}

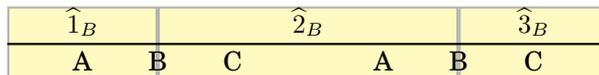
\begin{center}
\begin{pspicture}(80,15)
\agentspieceframe{B}{1}{0}{20}
\agentspieceframe{B}{2}{20}{60}
\agentspieceframe{B}{3}{60}{80}
\agentcuts{A}{B}{C}
\end{pspicture}
\end{center}
Ask Bob to \equalize(3). If Alice and Carl prefer different pieces, then we are done --- an envy-free allocation exists with the current 3 pieces, so by the EFP lemma the proportionality is 1/3. The branch "Bob:\equalize(3)" succeeds.

Otherwise, Alice and Carl prefer the same piece. This piece must be \agentspiece{B}{2}, since \agentspiece{B}{1} is worth less than 1 for Carl and \agentspiece{B}{3} is worth less than 1 for Alice. This means that \agentspiece{B}{2} is worth more than 1 for both Alice and Carl. Hence, each of them has two pieces worth at least 1: \agentspiece{B}{1} and \agentspiece{B}{2} for Alice, \agentspiece{B}{2} and \agentspiece{B}{3} for Carl. This is the same situation as in the case C-B-A. An \equalize(2) by either Alice or Carl guarantees an envy-free and proportional division. So the branch "Bob:\equalize(3); Alice:\equalize(2)" succeeds.

\subsection{B-A-C}

\begin{center}
\begin{pspicture}(80,15)
\agentspieceframe{B}{1}{0}{20}
\agentspieceframe{B}{2}{20}{50}
\agentspieceframe{B}{3}{50}{80}
\agentcuts{B}{A}{C}
\end{pspicture}
\end{center}
Ask Bob to \equalize(3). Alice values two pieces ---  \agentspiece{B}{1} and \agentspiece{B}{3} --- as at least 1. Ask Alice to \equalize(2). Alice still has two pieces with a value of at least 1. As for Carl, there are two cases: If Alice trimmed \agentspiece{B}{1}, then \agentspiece{B}{3} remains untouched; its value for Carl is more than 1. If Alice trimmed \agentspiece{B}{3}, then she must have trimmed at or to the left of the second A (since its value for her must be at least the value of \agentspiece{B}{1}, which is at least 1):
\begin{center}
	\begin{pspicture}(80,15)
	\pieceframe{Alice}{0}{20}
	\pieceframe{Bob}{20}{50}
	\pieceframe{Carl}{60}{80}
	\agentcuts{B}{A}{C}
	\end{pspicture}
\end{center}
Hence, the value of the trimmed piece for Carl is still at least 1.  So the branch "Bob:\equalize(3); Alice:\equalize(2)" succeeds.

\subsection{A-C-B}

\begin{center}
\begin{pspicture}(80,15)
\agentspieceframe{C}{1}{0}{30}
\agentspieceframe{C}{2}{30}{60}
\agentspieceframe{C}{3}{60}{80}
\agentcuts{A}{C}{B}
\end{pspicture}
\end{center}
The previous case, A-B-C-B-A-C, is symmetric to A-B-C-A-C-B. This can be seen by renaming the agents from A-B-C to B-C-A and reversing the order of lines. The branch "Carl:\equalize(3); Bob:\equalize(2)" succeeds.

\subsection{B-C-A}
The last sub-case is handled according to Alice's preferences --- whether she prefers \agentspiece{B}{1} (which contains \agentspiece{A}{1}) or \agentspiece{C}{3} (which contains \agentspiece{A}{3}). Note that Alice values both these pieces as at least 1. 

\textbf{If Alice prefers \agentspiece{B}{1}}, then ask Bob to \equalize(3). Alice values two pieces as at least 1. Ask her to \equalize(2).  For Carl there are two cases: If Alice trimmed \agentspiece{B}{1}, then \agentspiece{B}{3} remains untouched; its value for Carl is more than 1. If Alice trimmed \agentspiece{B}{3}, then she must have trimmed it at or to the left of the C mark, since she values \agentspiece{B}{1} more than \agentspiece{C}{3}:
\begin{center}
	\begin{pspicture}(80,15)
	\pieceframe{Alice}{0}{20}
	\pieceframe{Bob}{20}{50}
	\pieceframe{Carl}{60}{80}
	\agentcuts{B}{C}{A}
	\end{pspicture}
\end{center}
Hence, the value of the trimmed piece for Carl is still at least 1.  The branch "Bob:\equalize(3); Alice:\equalize(2)" succeeds.

\textbf{If Alice prefers \agentspiece{C}{3}}, then ask Carl to \equalize(3). Alice still values two pieces as at least 1. Ask her to \equalize(2). For Bob there are two cases: If Alice trimmed \agentspiece{C}{3}, then \agentspiece{C}{1} remains untouched; its value for Bob is more than 1. If Alice trimmed \agentspiece{C}{1}, then she must have trimmed it at or to the right of the B mark, since she values \agentspiece{C}{3} more than \agentspiece{B}{1}. Hence, the value of the trimmed piece for Bob is still at least 1. The branch "Carl:\equalize(3); Alice:\equalize(2)" succeeds.
\begin{center}
\begin{pspicture}(80,15)
\pieceframe{Bob}{0}{20}
\pieceframe{Carl}{30}{60}
\pieceframe{Alice}{60}{80}
\agentcuts{B}{C}{A}
\end{pspicture}
\end{center}
This completes the correctness proof of the 3-agents division algorithm.\footnote{Note that the proof did not use all 9 branches in the algorithm. This is because the proof arbitrarily named the agents A, B and C according to the order of their leftmost division line. If the agents' names are given, each of the 9 branches may be required.}
\end{proof}
We now count the number of mark-eval queries required by Algorithm \ref{alg:3-agents} in the worst case. In order to try all 9 execution branches, the algorithm has to do three different Equalize(3) statements, each of which requires 2 \emph{mark}s, and six Equalize(2) statements, each of which requires 2 \emph{eval}s (since there are three pieces on the table) and 1 \emph{mark}. To check whether a branch succeeds (and implement the envy-free allocation if it exists), the algorithm should ask the agent/s who have not touched the cake so far to evaluate the pieces. This requires 4 \emph{eval}s in the three short branches and 3 \emph{eval}s in the six long ones. All in all, at most $3*2+6*3+3*4+6*3=54$ queries are required. This completes the proof of our Theorem \ref{thm:3-agents}.

\section{Connected pieces and 4 agents}  \label{sec:4-agents}

\psset{colsep=0.5cm,rowsep=1cm}

Encouraged by the performance of the algorithm of Section \ref{sec:3-agents}, we would like to extend it to produce a connected envy-free and proportional allocation for $n$ agents. Unfortunately, the number of different cases becomes prohibitively large even for $n=4$ agents. The equal partition of each agent is made by 3 parallel marks, so if we name the agents according to their 1st mark, the number of options for the following two marks is $(4!)^2=576$, and in general $(n!)^{n-2}$. The algorithm for each specific case may be short, but writing down all the different cases takes too long to be practical.

\begin{algorithm}
\caption{Finding allocations which are both \envyfreevip(4,4) and \envyfree (4,7) with connected pieces.}
\label{alg:4-agents}

\dirtree{
.1 Alice:\equalize(4) .
.1 One of: .
.2 Bob:\equalize(2); Carl:\equalize(2) .
.2 Bob:\equalize(3); Carl:\equalize(2) .
.2 Carl:\equalize(2); Bob:\equalize(2) .
.2 Carl:\equalize(3); Bob:\equalize(2) .
}
\end{algorithm}

Therefore we walk in a different direction and present Algorithm \ref{alg:4-agents} for 4 agents, which we call Alice, Bob, Carl and Dana. The main guarantee of this algorithm is:

\begin{lemma}\label{lem:4-agents}
For every preferences of the agents, there is at least one branch of Algorithm \ref{alg:4-agents} in which, in the resulting preference graph, each of Alice Bob and Carl prefers at least 2 pieces.
\end{lemma}

By Lemma \ref{lem:unknown1}, this implies that 
Algorithm \ref{alg:4-agents} finds an envy-free allocation regardless of Dana's preferences. Since at least one of Alice's original pieces remains untouched, the resulting allocation is also \envyfreevip(4,4) (Alice is the VIP). The total number of cuts in each branch is at most 6 so the total number of pieces is at most 7. Hence the resulting allocation is \envyfree(4,7).

We now count the number of mark-eval queries required. The first Equalize(4) requires 3 \emph{mark}s. The second Equalize requires 3 \emph{eval}s (since there are 4 pieces on the table) and 1 or 2 \emph{marks}; the last Equalize requires 4 or 5 \emph{eval}s (depending on the number of pieces on the table) and 1 \emph{mark}. Checking whether a branch succeeds requires asking Dana 5 or 6 \emph{eval}s. All in all, the number of queries required in all paths is $3+(3+1+4+1+5)+(3+2+5+1+6)+(3+1+4+1+5)+(3+2+5+1+6)=65$, as claimed in our Theorem \ref{thm:4-agents}.

\begin{proof}[of Lemma \ref{lem:4-agents}]
After Alice:\equalize(4), there are four pieces on the table. We rename the pieces, if needed, such that Bob's preference ordering on these pieces is: $\piece{1} \preceq \piece{2} \preceq \piece{3} \preceq \piece{4}$. Now there are $4!=24$ possible preference orderings for Carl. Since checking 24 cases is a tedious task, we wrote a program in SageMath \citep{sage} to do it. Our program generates a textual proof that can be read and verified independently of the program itself (i.e, it is not required to believe that the program is bug-free in order to verify the proof). The entire proof is given in Appendix \ref{sec:proof-4-agents}. Below, we explain the typical cases in detail.

Mark piece \piece{i} after agent $X$ cuts it during \equalize(2) by $\piece{i}_X$ and during \equalize(3) by $\piece{i}_{XX}$. The following graphs show the preferences of Alice and Bob after Bob does \equalize(2) (left) or \equalize(3) (right):
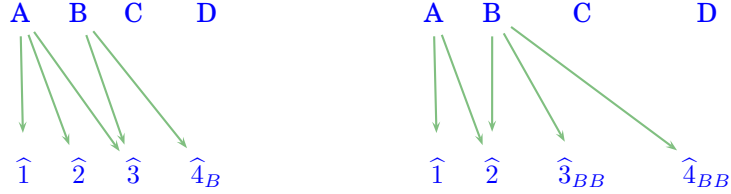
\begin{center}
\begin{psmatrix}[colsep=2cm,rowsep=0mm]
	\begin{psmatrix}[colsep=0cm,rowsep=0mm]
		\fouragents{}
        \piecenode{1} & \piecenode{2} & \piecenode{3} & \piecenodecut{4}{B}
        
		\like{A}{1}\like{A}{2}\like{A}{3}
		\like{B}{3}\like{B}{4B}
	\end{psmatrix}
& 
	\begin{psmatrix}[colsep=0cm,rowsep=0mm]
		\fouragents{}
        \piecenode{1} & \piecenode{2} & \piecenodecut{3}{BB} & \piecenodecut{4}{BB}
	
		\like{A}{1}\like{A}{2}
		\like{B}{2}\like{B}{3BB}\like{B}{4BB}
	\end{psmatrix}
\end{psmatrix}
\end{center}
Of the 24 possible preference orders of Carl, there are 6 cases in which the best piece of Carl is \piece{1}. This obviously remains Carl's best piece after Bob trims some other pieces. This means that when Carl does \equalize(2), piece \piece{1} is trimmed. Looking at the left graph, we see that afterwards, each of Alice Bob and Carl prefers at least two pieces. This means that the branch starting with "Bob:\equalize(2)" succeeds. The same is true for the 6 cases in which Carl's best piece is \piece{2}; we have already covered 12 out of 24 cases.

Next, consider the four cases in which Carl's best piece is \piece{3} and Carl's second-best piece is \piece{1} or \piece{2}. Then, after Bob:\equalize(2), Carl's best piece is \piece{3} and trimming it leaves only one best piece for Bob, so the branch starting with "Bob:\equalize(2)" fails. On the other hand, after Carl:\equalize(2), Bob's best piece is still \piece{4} and trimming it leaves two best pieces for Alice and Carl, so the branch starting with "Carl:\equalize(2)" succeeds. So far, 16 of 24 cases are covered.

Next, consider the two cases in which Carl's best piece is \piece{4} and second-best piece is \piece{1}. After Bob:\equalize(2), \piececut{4}{B} may or may not be Carl's best piece:
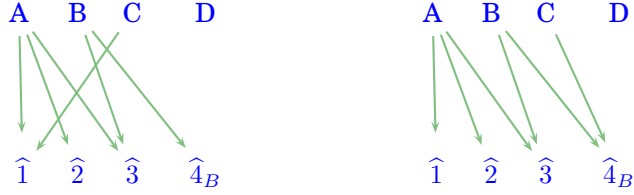
\begin{center}
\begin{psmatrix}[colsep=2cm,rowsep=0mm]
	\begin{psmatrix}[colsep=0cm,rowsep=0mm]
		\fouragents{}
        \piecenode{1} & \piecenode{2} & \piecenode{3} & \piecenodecut{4}{B}
	        
		\like{A}{1}\like{A}{2}\like{A}{3}
		\like{B}{3}\like{B}{4B}
		\like{C}{1}
	\end{psmatrix}
&
	\begin{psmatrix}[colsep=0cm,rowsep=0mm]
		\fouragents{}
        \piecenode{1} & \piecenode{2} & \piecenode{3} & \piecenodecut{4}{B}
	        
		\like{A}{1}\like{A}{2}\like{A}{3}
		\like{B}{3}\like{B}{4B}
		\like{C}{4B}
	\end{psmatrix}
\end{psmatrix}
\end{center}
The easy case is that for Carl, $\piececut{4}{B}\preceq\piece{1}\preceq\piece{4}$ (left); then, Carl:\equalize(2) trims \piece{1} and leaves two best pieces for Alice and Bob, and the branch starting with "Bob:\equalize(2)" succeeds. The hard case is  $\piece{1}\preceq\piececut{4}{B}\preceq\piece{4}$ (right); then, Carl:\equalize(2) trims \piececut{4}{B} and leaves only one best piece for Bob, and the branch starting with "Bob:\equalize(2)" fails. But now, consider the branch starting with "Carl:\equalize(2)". Carl trims piece \piece{4} to make it equal to \piece{1}. But because $\piece{1}\preceq\piececut{4}{B}\preceq\piece{4}$, Carl must trim \piece{4} to a \emph{shorter} length than \piececut{4}{B} (we assume that all agents trim the pieces from the same direction), so $ \piececut{4}{C} \subseteq \piececut{4}{B}$. This means that, if the branch "Bob:\equalize(2)" fails, the following preference relation must be true globally (for \emph{all} agents):

$$\piececut{4}{C} \preceq \piececut{4}{B}$$

In particular, it must be true for Bob. But by definition, for Bob, $\piececut{4}{B}$ and $\piece{3}$ are equal. Hence, for Bob: $\piececut{4}{C} \preceq \piece{3}$. 
This means that, after the initial Carl:\equalize(2), Bob's best piece is \piece{3}:
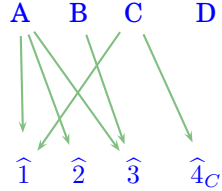
\begin{center}
\begin{psmatrix}[colsep=2cm,rowsep=0mm]
	\begin{psmatrix}[colsep=0cm,rowsep=0mm]
		\fouragents{}
        \piecenode{1} & \piecenode{2} & \piecenode{3} & \piecenodecut{4}{C}
	        
		\like{A}{1}\like{A}{2}\like{A}{3}
		\like{B}{3}
		\like{C}{1}\like{C}{4C}
	\end{psmatrix}
\end{psmatrix}
\end{center}
Now, Bob:\equalize(2) leaves two best pieces for both Alice and Carl, so the branch starting with "Carl:\equalize(2)" succeeds.

An almost identical argument applies in the two cases in which Carl's best piece is \piece{4} and second-best piece is \piece{2}, so we already covered 20 of 24 cases; in each of these cases, either the branch starting with "B:\equalize (2)" or the branch starting with "C:\equalize (2)" succeeds.

The other two branches are used in the remaining four cases: in each of these cases, either the branch starting with "B:\equalize(3)" or the branch starting with "C:\equalize(3)" must succeed. The proof uses similar arguments to the one explained above: if the branch starting with "B:\equalize(3)" fails, then some global containment relations are implied between \piececut{4}{BB} and \piececut{4}{CC} and between \piececut{3}{BB} and \piececut{3}{CC}. These relations imply that the branch starting with "C:\equalize(3)" must succeed. Although the proof is longer, the principle is the same so we leave the details to the printout in Appendix \ref{sec:proof-4-agents}.
\end{proof}

Lemma \ref{lem:4-agents} lets us improve the algorithm of Section \ref{sec:n-agents}. We note that the core of that algorithm is Lemma \ref{lem:unknownn}, which is exponential in nature --- it requires that every known agent prefers $1+2^{u-1}$ pieces whenever there are $u$ known agents. We would like to reduce this figure to $1+u$. Since for $u\in\{1,2\}$ these two expressions are equal, we focus on the case $u=3$:

\begin{lemma}
\label{lem:unknown3}
Suppose the $n$ agents are divided to $n-3$ agents whose preferences are known and 3 agent whose preferences are unknown. If every known agent prefers at least 4 pieces, then an envy-free allocation exists.
\end{lemma}
\begin{proof}
Call the three unknown agents Bob, Carl and Dana. Apply Algorithm \ref{alg:4-agents} without the first step "Alice:\equalize(4)". The assumption of the lemma implies that every known agent is in the same situation as Alice after the first step. This means that after the algorithm completes, all agents except the last unknown agent prefer two pieces. By Lemma \ref{lem:unknown1}, an envy-free allocation exists.
\end{proof}

Lemma \ref{lem:unknown3} can be plugged as a base case into Lemma \ref{lem:unknownn} to get the following improved lemma:
\begin{lemma}
\label{lem:unknownn3}
Suppose the $n$ agents are divided to $n-u$ agents whose preferences are known and $u$ agents whose preferences are unknown, where $u\geq 3$. If every known agent prefers at least $1+3\cdot 2^{u-3}$ pieces, then an envy-free allocation can be attained with a bounded number of queries.
\end{lemma}
The proportionality of the division algorithm of Section \ref{sec:n-agents} improves to $1/[\frac{3}{4}\cdot 2^{n-1}+1]$.


\begin{remark}
The symmetry of Algorithm \ref{alg:4-agents} hints that it may be generalizable to 5 or more agents. In particular, we thought that an algorithm such as the following might work:

\begin{framed}
Alice:\equalize(5)
\dirtree{
	.1 One of: .
	.2 Bob:\equalize(2); Carl:\equalize(2); Dana:\equalize(2) .
	.2 Bob:\equalize(2); Carl:\equalize(3); Dana:\equalize(2) .
	.2 Bob:\equalize(3); Carl:\equalize(2); Dana:\equalize(2) .
	.2 Bob:\equalize(3); Carl:\equalize(3); Dana:\equalize(2) .
	.2 Bob:\equalize(4); Carl:\equalize(2); Dana:\equalize(2) .
	.2 Bob:\equalize(4); Carl:\equalize(3); Dana:\equalize(2) .
	.2 [and similarly for the other 5 permutations of Bob, Carl, Dana] .
}
\end{framed}

We checked this possibility using our SageMath program, but found a specific combination of preferences in which all these branches fail to produce a saturated matching.\footnote{One such case is when Bob's preference order is $1\preceq 2\preceq 3\preceq 4\preceq 5$, Carl's order is also $1\preceq 2\preceq 3\preceq 4\preceq 5$ while Dana's order is $1\preceq 3\preceq 2\preceq 4\preceq 5$. See http://github.com/erelsgl/envy-free for the source code and proof.}
\end{remark}

\section{Disconnected pieces and 4 agents} \label{sec:disconnected-4}
In this and the following section, we use our results from the connected case to prove better proportionality bounds in the disconnected case. We show that, if the pieces may be disconnected, we can have an envy-free division in which the value of each agent is arbitrarily close to $1/n$, in bounded time. This is done using two general reduction lemmas which rely on envy-free-VIP algorithms. 

\begin{algorithm}
\caption{Finding \envyfree (4,4) allocations with disconnected pieces.}
\label{alg:4-agents-disc}

Let $C' = C$.

For $i$ = 1 to 4:

\hspace{1cm} Rename \agent{i} to "Alice"; 

\hspace{1cm} Divide $C'$ using Algorithm \ref{alg:4-agents};

\hspace{1cm} Let $C' = $ the subset that  remained unallocated.
\end{algorithm}

\begin{lemma} (Weak Reduction Lemma)
For every $n$ and $M\geq n$,

If there is an algorithm for finding $\envyfreevip(n,M)$ allocations using $T(n)$ queries,

then there is an algorithm for finding $\envyfree(n,M)$ allocations using $n\cdot T(n)$ queries.
\end{lemma}
\begin{proof} (generalizing an idea of \citet{Saberi2009Cutting}). The idea is to use the existing algorithm $n$ times, each time on the remainder of the previous time and with a different agent as the VIP. This ensures that all agents enjoy the VIP proportion of $1/M$.

Let $C$ be the original cake. Run $\envyfreevip(n,M)$ on $C$ with agent \agent{1} as the VIP. The result is an allocation of a certain subset of $C$ (say, $C'\subseteq C$) with the following properties:
\begin{itemize}
\item The allocation of $C'$ is envy-free.
\item Hence, by the pigeonhole principle, every agent \agent{i} has a value of at least $V_i(C')/n\geq V_i(C')/M$.
\item Moreover, \agent{1} holds a value of at least $V_1(C)/M$.
\end{itemize}

If $C'=C$, then we are done since every agent \agent{i} holds a value of at least $V_i(C)/M$. Otherwise, there is a remainder, $\overline{C'}=C\setminus C'$, that should be divided. Run $\envyfreevip(n,M)$ on that remainder with \agent{2} as the VIP. The result is an allocation of a certain subset $C''\subseteq \overline{C'}$ with the following properties:
\begin{itemize}
\item The allocation of $C''$ is envy-free.
\item Hence, by the pigeonhole principle, every agent \agent{i} holds a value of at least $V_i(C'')/n \geq V_i(C'')/M$.
\item Moreover, \agent{2} has a value of at least $V_2(\overline{C'})/M$.
\end{itemize}

Combining the two previous allocations, we now have an allocation of $C' \cup C''$, with the following properties:

\begin{itemize}
\item The allocation of $C' \cup C''$ is envy-free (since it is a combination of two envy-free allocations).
\item Hence, by the pigeonhole principle, every \agent{i} has a value of at least $V_i(C'\cup C'')/n \geq V_i(C'\cup C'')/M$.
\item \agent{1} (still) has a value of at least $V_1(C)/M$, since nothing was taken from him.
\item \agent{2} has a value of at least $V_2(C')/M + V_2(\overline{C'})/M$, which is at least $V_2(C)/M$.
\end{itemize}

So after the second division, we have an envy-free division in which both \agent{1} and \agent{2} hold at least $1/M$ of their total cake value.

If $C'\cup C''=C$ then we are done. Otherwise, there is a remainder $\overline{C'\cup C''}$ that should be divided. Continue in the same way: run $\envyfreevip(n,M)$ on that remainder with agent \agent{3} as the VIP, then with \agent{4} as the VIP, and so on. It is easy to prove by induction that, after at most $n$ runs, all agents have at least $1/M$ of their total cake value.
\end{proof}

The Weak Reduction Lemma is most useful in the case $M=n$. Note that an $\envyfree(n,n)$ allocation is both envy-free \emph{and proportional}. The Weak Reduction Lemma implies that such an allocation can be found using an algorithm which guarantees a value of at least $1/n$ to a \emph{single} VIP agent. 

The 4-agent algorithm of Section \ref{sec:4-agents} finds an $\envyfreevip(4,4)$ allocation using at most 6 cuts and 65 queries, so an $\envyfree(4,4)$ allocation can be found using at most $6*4=24$ cuts and $65*4=260$ queries. This proves our Theorem \ref{thm:4-agents-disc}.

\section{Disconnected pieces and \MakeLowercase{\LARGE {$n$}} agents} \label{sec:disconnected-n}
The following lemma allows us to approach an $\envyfreevip(n,n)$ allocation to any desired precision.

\begin{algorithm}
\caption{Finding \envyfreevip(n,\frac{n}{1-\epsilon}) allocations with disconnected pieces.}
\label{alg:n-agents-disc}
Rename the agents such that the VIP is \agent{1}.

Let $C' = C$.

For $t$ = 1 to $\lceil (2^{n-2}+1)\ln{(1/\epsilon)}/n \rceil$:

\hspace{1cm} Divide $C'$ using Algorithm \ref{alg:n-agents};

\hspace{1cm} Let $C' = $ the subset that  remained unallocated.
\end{algorithm}

\begin{lemma} (Strong Reduction Lemma)
For every $n$, $M>n$ and $\epsilon>0$:

If there is an algorithm for finding $\envyfreevip(n,M)$ allocations using $T(n)$ queries,

then there is an algorithm for finding $\envyfreevip(n, \frac{n}{1-\epsilon})$ allocations using $\factor \cdot T(n)$ queries.
\end{lemma}
\begin{proof}
The main idea is to use the existing algorithm many times, each time on the remainder of the previous time, with the same agent as the VIP. The value of the VIP agent grows like a geometric series and converges to $1/n$. Hence, after a sufficient number of runs, the VIP agent's value is at least $(1-\epsilon)/n$.

The proof uses the following notation:
\begin{itemize}
\item $t$ --- the number of times the $\envyfreevip(n,M)$ algorithm has been run on successive remainders. 
\item $C'_t$ ($t\geq 1$) --- the part of $C$ allocated at time $t$. 
\item $C_t (t\geq 0)$ --- the total cake allocated up to and including time $t$ ($C_t :=\cup_{j=1}^t{C'_j}$).
\item $V'_t$ ($t\geq 1$) --- the value given to the VIP agent at time $t$.
\item $V_t$ ($t\geq 0$) --- the total value held by the VIP agent after time $t$ ($V_t:=\sum_{j=1}^t{V'_j}$).
\end{itemize}

We first prove that, in every time $t\geq 1$:
\begin{align}
\label{claim:vttag}
	V'_t\geq [1-n V_{t-1}]/M
\end{align}
\begin{proof}
Since all allocations are envy-free, the cumulative allocation of $C_{t-1}$ is also envy-free. This means that the VIP agent, like all other agents, holds at least a proportional share of it:  $$V_{t-1} \geq V(C_{t-1})/n$$
At time $t$, the cake that remains to be divided is $C\setminus C_{t-1}$. The VIP agent receives at least a fraction $1/M$ of it: $$V'_t\geq [1-V(C_{t-1})]/M$$
Combining the previous two inequalities gives the desired inequality.
\end{proof}

Next we prove that in every time $t\geq 0$:
\begin{align}
\label{claim:vt}
V_t \geq \frac{1-(1-n/M)^t}{n}
\end{align}
\begin{proof}
By induction on $t$. For $t=0$, by definition $V_0=0$. Suppose the claim is true for $t$, so there is a constant $d \geq 0$ such that:
\[
V_t = \frac{1-(1-n/M)^t}{n} + d
\]
By inequality (\ref{claim:vttag}):
\[
V'_{t+1} \geq [1 - n V_t]/M
\]
By the induction assumption, $1-n V_t = (1-n/M)^t-nd$. Hence:
\[
V'_{t+1} \geq [(1-n/M)^t]/M - nd/M 
\]
So: 
\[
V_{t+1} = V_t + V'_{t+1} \geq \frac{1-(1-n/M)^t+(n/M)(1-n/M)^t}{n} + d(1-n/M)
\]
Because $M\geq n$, the rightmost term is positive and we get the desired inequality:
\[
V_{t+1} \geq \frac{1-(1-n/M)^{t+1}}{n}
\]
\end{proof}

By inequality (\ref{claim:vt}), to get a value of at least $V_t\geq (1-\epsilon)/n$, it is sufficient to choose $t$ such that:
\[
(1-n/M)^t \leq \epsilon
\]
The latter inequality is true whenever:
\[
t \geq \frac{\ln{\epsilon}}{\ln(1-n/M)} = \frac{\ln(1/\epsilon)}{-\ln(1-n/M)}
\]
By the log inequality: $\ln(1-n/M)<-n/M$, so it is sufficient that:
\[
t\geq \frac{\ln{(1/\epsilon)}}{n/M} = \frac{M\ln{(1/\epsilon)}}{n}
\]
So it is sufficient to run our $\envyfreevip(n,M)$ algorithm $\factor$ times. 
\qed
\end{proof}

By combining the two reduction lemmas we get:
\begin{corollary}
For every $n$, $M>n$ and $\epsilon>0$:

If there is an algorithm for finding $\envyfreevip(n,M)$ allocations using $T(n)$ queries,

then there is an algorithm for finding $\envyfree(n, \frac{n}{1-\epsilon})$ allocations using $n\cdot\factor \cdot T(n) \approx M\ln(1/\epsilon) \cdot T(n)$ queries.
\end{corollary}

In the algorithm of Section \ref{sec:n-agents}, the first cutter cuts $2^{n-2}+1$ equal pieces. Hence, the resulting allocations are  $\envyfreevip(n,2^{n-2}+1)$. The total number of queries is $O(2^n)$. Hence:

\begin{corollary}
For every $n$ and $\epsilon>0$, there exists an algorithm for finding $\envyfree(n, \frac{n}{1-\epsilon})$ allocations using $O(4^n\ln(1/\epsilon))$ queries.
\end{corollary}

This completes the proof of our Theorem \ref{thm:n-agents-disc}.

\section{Future Work}
The main question left open by the present paper is:
\begin{quote}
	Is there a bounded-time algorithm for finding an envy-free and proportional allocation with connected pieces for 4 or more agents?
\end{quote}
The most recent advancement we are aware of was made by Aziz and Mackenzie at 29/7/2016.\footnote{http://arxiv.org/abs/1604.03655v7} They developed an algorithm that finds an envy-free allocation with connected pieces in which each agent receives a value of at least $1/(3n)$.

\appendix
\section*{APPENDIX}
\setcounter{section}{0}

\section{Automatically-generated proof for 4-agent envy-free-VIP algorithm} \label{sec:proof-4-agents}

For convenience, we repeat here the algorithm of Section \ref{sec:4-agents}:

\begin{framed}
Alice:\equalize(4)
\dirtree{
.1 One of: .
.2 Bob:\equalize(2); Carl:\equalize(2) .
.2 Bob:\equalize(3); Carl:\equalize(2) .
.2 Carl:\equalize(2); Bob:\equalize(2) .
.2 Carl:\equalize(3); Bob:\equalize(2) .
}
\end{framed}

Below, when the proof says e.g. that "B:\equalize(2)... always succeeds", it means that the execution branch starting with "Bob:\equalize(2)" necessarily results in a preference graph in which Alice, Bob and Carl each prefer two pieces. Then, by Lemma \ref{lem:unknown1} an envy-free allocation exists. Conversely, when the proof says that "this must fail", it means that the execution branch necessarily does not result in such a preference graph. The proof systematically checks all possible preference relations and proves that in all cases, at least one of the four execution branches must succeed. The notation is explained in Section \ref{sec:4-agents}. E.g, "4B" is what remains of piece 4 after it is cut by Bob doing \equalize(2), "3CC" is what remains of piece 3 after it is cut by Carl doing \equalize(3), etc.

The source code of the program used to generate the proof is available in https://github.com/erelsgl/envy-free . We emphasize that the proof can be read and verified without the source code, so the correctness of the proof does not depend on the correctness of the code.

\begin{scriptsize}
\begin{verbatim}

Initially, agent A cuts four equal pieces:  1,2,3,4 .
Assume w.l.o.g. that B's preferences are 1<2<3<4 .
Consider the following 24 cases regarding the preferences of C:

CASE 1 OF 24 : C's order is 4<3<2<1 :
  B:Equalize(2) makes B's best pieces: 3=4B. This always succeeds.

CASE 2 OF 24 : C's order is 4<3<1<2 :
  B:Equalize(2) makes B's best pieces: 3=4B. This always succeeds.

CASE 3 OF 24 : C's order is 4<2<3<1 :
  B:Equalize(2) makes B's best pieces: 3=4B. This always succeeds.

CASE 4 OF 24 : C's order is 4<2<1<3 :
  B:Equalize(2) makes B's best pieces: 3=4B. This must fail because of C.
    C:Equalize(2) makes C's best pieces: 1=3C. This always succeeds.

CASE 5 OF 24 : C's order is 4<1<3<2 :
  B:Equalize(2) makes B's best pieces: 3=4B. This always succeeds.

CASE 6 OF 24 : C's order is 4<1<2<3 :
  B:Equalize(2) makes B's best pieces: 3=4B. This must fail because of C.
    C:Equalize(2) makes C's best pieces: 2=3C. This always succeeds.

CASE 7 OF 24 : C's order is 3<4<2<1 :
  B:Equalize(2) makes B's best pieces: 3=4B. This always succeeds.

CASE 8 OF 24 : C's order is 3<4<1<2 :
  B:Equalize(2) makes B's best pieces: 3=4B. This always succeeds.

CASE 9 OF 24 : C's order is 3<2<4<1 :
  B:Equalize(2) makes B's best pieces: 3=4B. This always succeeds.

CASE 10 OF 24 : C's order is 3<2<1<4 :
  B:Equalize(2) makes B's best pieces: 3=4B. This may fail in 1 case : C prefers 4B to 1 2 3 .
   Assume the case   C prefers 4B to 1 2 3. Then:
    C:Equalize(2) makes C's best pieces: 1=4C, so globally: 4C<4B . This always succeeds.

CASE 11 OF 24 : C's order is 3<1<4<2 :
  B:Equalize(2) makes B's best pieces: 3=4B. This always succeeds.

CASE 12 OF 24 : C's order is 3<1<2<4 :
  B:Equalize(2) makes B's best pieces: 3=4B. This may fail in 1 case : C prefers 4B to 1 2 3 .
   Assume the case   C prefers 4B to 1 2 3. Then:
    C:Equalize(2) makes C's best pieces: 2=4C, so globally: 4C<4B . This always succeeds.

CASE 13 OF 24 : C's order is 2<4<3<1 :
  B:Equalize(2) makes B's best pieces: 3=4B. This always succeeds.

CASE 14 OF 24 : C's order is 2<4<1<3 :
  B:Equalize(2) makes B's best pieces: 3=4B. This must fail because of C.
    C:Equalize(2) makes C's best pieces: 1=3C. This always succeeds.

CASE 15 OF 24 : C's order is 2<3<4<1 :
  B:Equalize(2) makes B's best pieces: 3=4B. This always succeeds.

CASE 16 OF 24 : C's order is 2<3<1<4 :
  B:Equalize(2) makes B's best pieces: 3=4B. This may fail in 1 case : C prefers 4B to 1 2 3 .
   Assume the case   C prefers 4B to 1 2 3. Then:
    C:Equalize(2) makes C's best pieces: 1=4C, so globally: 4C<4B . This always succeeds.

CASE 17 OF 24 : C's order is 2<1<4<3 :
  B:Equalize(2) makes B's best pieces: 3=4B. This must fail because of C.
    C:Equalize(2) makes C's best pieces: 4=3C. This must fail because of B.
      B:Equalize(3) makes B's best pieces: 2=3BB=4BB. This may fail in 1 case : C prefers 1 to 2 3BB 4BB .
       Assume the case   C prefers 1 to 2 3BB 4BB. Then:
        C:Equalize(3) makes C's best pieces: 1=4CC=3CC, so globally: 4BB<4CC 3BB<3CC . This always succeeds.

CASE 18 OF 24 : C's order is 2<1<3<4 :
  B:Equalize(2) makes B's best pieces: 3=4B. This may fail in 2 cases : C prefers 4B to 1 2 3;  C prefers 3 to 1 2 4B .
   Assume the case   C prefers 4B to 1 2 3. Then:
    C:Equalize(2) makes C's best pieces: 3=4C, so globally: 4C<4B . This may fail in 1 case : B prefers 3 to 2 1 4C .
     Assume the case   B prefers 3 to 2 1 4C. Then:
      B:Equalize(3) makes B's best pieces: 2=3BB=4BB. This may fail in 1 case : C prefers 1 to 2 3BB 4BB .
       Assume the case   C prefers 1 to 2 3BB 4BB. Then:
        C:Equalize(3) makes C's best pieces: 1=3CC=4CC, so globally: 3BB<3CC 4BB<4CC . This always succeeds.
   Assume the case   C prefers 3 to 1 2 4B. Then:
    C:Equalize(2) makes C's best pieces: 3=4C, so globally: 4B<4C . This may fail in 1 case : B prefers 4C to 2 1 3 .
     Assume the case   B prefers 4C to 2 1 3. Then:
      B:Equalize(3) makes B's best pieces: 2=3BB=4BB. This may fail in 1 case : C prefers 1 to 2 3BB 4BB .
       Assume the case   C prefers 1 to 2 3BB 4BB. Then:
        C:Equalize(3) makes C's best pieces: 1=3CC=4CC, so globally: 3BB<3CC 4BB<4CC . This always succeeds.

CASE 19 OF 24 : C's order is 1<4<3<2 :
  B:Equalize(2) makes B's best pieces: 3=4B. This always succeeds.

CASE 20 OF 24 : C's order is 1<4<2<3 :
  B:Equalize(2) makes B's best pieces: 3=4B. This must fail because of C.
    C:Equalize(2) makes C's best pieces: 2=3C. This always succeeds.

CASE 21 OF 24 : C's order is 1<3<4<2 :
  B:Equalize(2) makes B's best pieces: 3=4B. This always succeeds.

CASE 22 OF 24 : C's order is 1<3<2<4 :
  B:Equalize(2) makes B's best pieces: 3=4B. This may fail in 1 case : C prefers 4B to 1 2 3 .
   Assume the case   C prefers 4B to 1 2 3. Then:
    C:Equalize(2) makes C's best pieces: 2=4C, so globally: 4C<4B . This always succeeds.

CASE 23 OF 24 : C's order is 1<2<4<3 :
  B:Equalize(2) makes B's best pieces: 3=4B. This must fail because of C.
    C:Equalize(2) makes C's best pieces: 4=3C. This must fail because of B.
      B:Equalize(3) makes B's best pieces: 2=3BB=4BB. This may fail in 1 case : C prefers 2 to 1 3BB 4BB .
       Assume the case   C prefers 2 to 1 3BB 4BB. Then:
        C:Equalize(3) makes C's best pieces: 2=4CC=3CC, so globally: 4BB<4CC 3BB<3CC . This always succeeds.

CASE 24 OF 24 : C's order is 1<2<3<4 :
  B:Equalize(2) makes B's best pieces: 3=4B. This may fail in 2 cases : C prefers 4B to 1 2 3;  C prefers 3 to 1 2 4B .
   Assume the case   C prefers 4B to 1 2 3. Then:
    C:Equalize(2) makes C's best pieces: 3=4C, so globally: 4C<4B . This may fail in 1 case : B prefers 3 to 1 2 4C .
     Assume the case   B prefers 3 to 1 2 4C. Then:
      B:Equalize(3) makes B's best pieces: 2=3BB=4BB. This may fail in 1 case : C prefers 2 to 1 3BB 4BB .
       Assume the case   C prefers 2 to 1 3BB 4BB. Then:
        C:Equalize(3) makes C's best pieces: 2=3CC=4CC, so globally: 3BB<3CC 4BB<4CC . This always succeeds.
   Assume the case   C prefers 3 to 1 2 4B. Then:
    C:Equalize(2) makes C's best pieces: 3=4C, so globally: 4B<4C . This may fail in 1 case : B prefers 4C to 1 2 3 .
     Assume the case   B prefers 4C to 1 2 3. Then:
      B:Equalize(3) makes B's best pieces: 2=3BB=4BB. This may fail in 1 case : C prefers 2 to 1 3BB 4BB .
       Assume the case   C prefers 2 to 1 3BB 4BB. Then:
        C:Equalize(3) makes C's best pieces: 2=3CC=4CC, so globally: 3BB<3CC 4BB<4CC . This always succeeds.

Q.E.D!
\end{verbatim}

\end{scriptsize}

\section{Envy-free division of an entire cake} \label{sec:entire-cake}
Recently, \citet{Aziz2015Discrete} have made an important breakthrough in the search for bounded-time envy-free cake-cutting algorithms. They presented the first bounded-time algorithm for envy-free cake-cutting of an \emph{entire cake} to 4 agents. In this appendix, we use our envy-free-VIP algorithm of Section \ref{sec:4-agents} to present their results in a simpler and more general way. 

\subsection{The domination graph}
The main new concept required for envy-free division of an entire cake is \textbf{domination}.\footnote{This concept originated with \citet{Brams1996Fair}, who used the term "irrevocable advantage". Aziz and Mackenzie introduced the shorter term "domination".} We say that Alice dominates Bob if Alice won't envy Bob even if the entire remaining cake is given to Bob. 

To see how a domination relation is created, consider again the envy-free-VIP algorithm for three agents, presented in Subsection \ref{sub:equalize}. Alice does \equalize(3) and Bob does \equalize(2), cutting his favorite piece, say \piece{3}, to make it equal to his second-best, say \piece{2}. Mark the trimmed piece \agentspiece{B}{3} and the trimmings \piece{4}, so that $\piece{3} = \piececut{3}{B} \cup \piece{4}$. Suppose Carl's best piece is \piece{2}. So Carl takes \piece{2}, Bob takes \agentspiece{B}{3}, Alice takes \piece{1}, and \piece{4} remains for the next round. Now, the following equalities hold for Alice:

$$V_A(\piece{1}) = V_A(\piece{2}) = V_A(\piece{3}) = V_A(\agentspiece{B}{3})+V_A(\piece{4})$$
Hence, even if the entire remainder (\piece{4}) is given to Bob, Alice will not envy. This means that Alice dominates Bob.

The domination relation can be described by a \textbf{domination graph}. In a domination graph, the nodes are the agents and an edge between two agents means that the source node dominates the target node. 

\begin{figure}
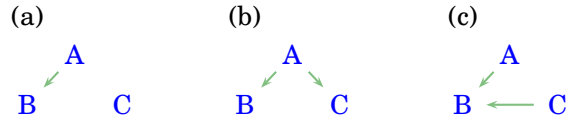

	\begin{center}
	\begin{psmatrix}[colsep=1cm,rowsep=0mm]
		\begin{psmatrix}[colsep=0cm,rowsep=0mm]
			(a) \\
			\threeagentsdomination{}
			\like{A}{B}
		\end{psmatrix}
		& 
		\begin{psmatrix}[colsep=0cm,rowsep=0mm]
			(b) \\
			\threeagentsdomination{}
			\like{A}{B}\like{A}{C}
		\end{psmatrix}
		& 
		\begin{psmatrix}[colsep=0cm,rowsep=0mm]
			(c) \\
			\threeagentsdomination{}
			\like{A}{B}\like{C}{B}
		\end{psmatrix}
	\end{psmatrix}
	\end{center}
	\caption{Domination graphs with 3 agents. \label{fig:domination-3}}
\end{figure}

Three domination graphs are shown in Figure \ref{fig:domination-3}. The graph (a) is generated after a single run of the algorithm, in which Alice dominates Bob (as explained above). Graph (b) can possibly occur after a second run of the same algorithm, if in the second run Carl takes the trimmed piece so Alice dominates Carl too. Graph (c) can occur after a second run of the algorithm in which Carl is the cutter, if Bob takes the trimmed piece. 

In general, when an envy-free algorithm is repeatedly executed, each time on the remainder of the previous time, edges are added to the domination graph but never removed. 

\subsection{Solvable domination graphs}
A convenient approach to the envy-free division problem is to reduce a given instance to a simpler instance that we already know how to solve. Formally:

\begin{definition}
A domination graph of an envy-free cake-cutting problem for $n$ agents is called \textbf{solvable} if, once the state of the division arrives at that domination graph, the problem can be reduced to one or more envy-free cake-cutting problems for less than $n$ agents.
\end{definition}
The domination graphs in Figure \ref{fig:domination-3} (b) and (c) are solvable: in (b), the problem can be reduced to a 2-agent division between Bob and Carl, since Alice dominates both of them; in (c), the entire remainder can be given to Bob, since he is dominated by both Alice and Carl. In general:

\begin{lemma}
\label{lem:solvable-separation}
If, in a domination graph for $n$ agents, there is a partition of the agents to two nonempty groups such that every agent in group \#2 dominates all agents in group \#1, then the domination graph is solvable.
\end{lemma}
\begin{proof}
An envy-free division of the entire cake can be found by letting the agents in group \#1 (whose number is less than $n$) divide the remainder among them in an envy-free way.
\end{proof}

\begin{figure}
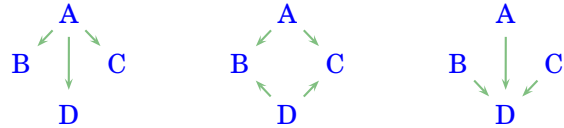

	\begin{center}
	\begin{psmatrix}[colsep=1cm,rowsep=0mm]
		\begin{psmatrix}[colsep=0cm,rowsep=0mm]
			\fouragentsdomination{}
			\like{A}{B}\like{A}{C}\like{A}{D}
		\end{psmatrix}
		& 
		\begin{psmatrix}[colsep=0cm,rowsep=0mm]
			\fouragentsdomination{}
			\like{A}{B}\like{A}{C}
			\like{D}{B}\like{D}{C}
		\end{psmatrix}
		& 
		\begin{psmatrix}[colsep=0cm,rowsep=0mm]
			\fouragentsdomination{}
			\like{A}{D}\like{B}{D}\like{C}{D}
		\end{psmatrix}
	\end{psmatrix}
	\end{center}
	\caption{Domination graphs with 4 agents. All are solvable by Lemma \ref{lem:solvable-separation}. \label{fig:domination-4}}
\end{figure}
Figure \ref{fig:domination-4} shows three graphs for 4 agents that are solvable by Lemma \ref{lem:solvable-separation}.

Another kind of solvable domination graphs is described in the following lemma.

\begin{lemma}
\label{lem:solvable-sequence}
If there is a sequence of $n-1$ agents, $A_2,\dots,A_{n}$, such that every agent dominates the following agents (every agent $A_i$ dominates every agent $A_j$ for all $2\leq i<j$), then the domination graph is solvable.
\end{lemma}
\begin{proof}
The remaining cake can be divided in the following way: $A_1$ (the agent not in the sequence) cuts the cake to $n$ equal parts. Then, the agents take pieces in the order $A_n,\dots,A_{2},A_1$. The agents in the sequence are not envious, because every agent dominates the agents that took pieces before him, and prefers his piece to the pieces taken by agents after him. $A_1$ is also not envious because all pieces are equal in his eyes.
\end{proof}

\begin{figure}
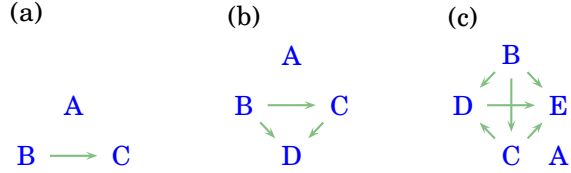

	\begin{center}
	\begin{psmatrix}[colsep=1cm,rowsep=0mm]
		\begin{psmatrix}[colsep=0cm,rowsep=0mm]
			(a) \\ \\ \\
			\threeagentsdomination{}
			\like{B}{C}
		\end{psmatrix}
		& 
		\begin{psmatrix}[colsep=0cm,rowsep=0mm]
			(b) \\
			\fouragentsdomination{}
			\like{B}{C}\like{B}{D}
			\like{C}{D}
		\end{psmatrix}
		& 
		\begin{psmatrix}[colsep=0cm,rowsep=0mm]
			(c) \\
			\fiveagentsdomination{}
			\like{B}{C}\like{B}{D}\like{B}{E}
			\like{C}{D}\like{C}{E}
			\like{D}{E}
		\end{psmatrix}
	\end{psmatrix}
	\end{center}
	\caption{Domination graphs solvable by Lemma \ref{lem:solvable-sequence}. \label{fig:domination-345}}
\end{figure}

Figure \ref{fig:domination-345} shows graphs for 3, 4 and 5 agents, that are solvable by Lemma \ref{lem:solvable-sequence}. In (a), the sequence is \{B,C\}, and the solution of Lemma \ref{lem:solvable-sequence} yields the well-known Selfridge-Conway algorithm. In (b) the sequence is \{B,C,D\} and in (c) it is \{B,C,D,E\}.

Combining the two previous lemmas gives a stronger lemma:

\begin{lemma}
\label{lem:solvable-n-2}
If there is a set of $n-1$ agents, each of whom dominates $n-2$ agents, then the domination graph is solvable.
\end{lemma}
\begin{proof}
Suppose that each of the agents $A_2,\dots,A_{n}$ dominates $n-2$ agents. Consider the following two cases:

Case \#1: all agents $A_2,\dots,A_{n}$ dominate $A_1$. Then by Lemma \ref{lem:solvable-separation} the domination graph is solvable.

Case \#2: there is an agent in $A_2,\dots,A_{n}$, say $A_2$, that does not dominate $A_1$. Hence, this agent must dominate all the other agents $A_3,\dots,A_{n}$. 

Regarding this smaller set of agents, there are again two cases:

Case \#2.1: all agents $A_3,\dots,A_{n}$ dominate both $A_1$ and $A_2$. Then by Lemma \ref{lem:solvable-separation} the domination graph is solvable.

Case \#2.2: there is an agent in $A_3,\dots,A_{n}$, say $A_3$, that does not dominate $A_1$ or does not dominate $A_2$. Hence, this agent must dominate all the other agents $A_4,\dots,A_{n}$. 

Regarding this smaller set of agents, we can continue the same line of reasoning. Finally we conclude that, either the domination graph is solvable by Lemma \ref{lem:solvable-separation}, or there exists a sequence of agents ($A_2,\dots,A_n$) such that each agent dominates the following agents in that sequence; then by Lemma \ref{lem:solvable-sequence} the domination graph is solvable.
\end{proof}

Lemma \ref{lem:solvable-n-2} implies that the problem of envy-free division of an entire cake among $n$ agents can be reduced to the following problem: 

\begin{quote}
Find an envy-free allocation of a part of a cake, such that a pre-specified VIP agent dominates $n-2$ agents.
\end{quote}
We now show that this reduced problem can be solved for $n=4$ agents.

\subsection{The $Equalize^*$ query}
First, we want to guarantee that after \emph{every} run of an envy-free-VIP algorithm, the VIP agent (the cutter) will dominate one of the agents. In order to guarantee this, we must change the semantics of the Equalize query. We call the changed query $Equalize^*$. An $Equalize^*(k)$ query asks an agent to cut his best $k-1$ pieces, such that the trimmed pieces will be equivalent to the agent's $k$-th best piece.   For example, an $Equalize^*(2)$ query to Bob in the above example implies the following question: "where would you cut piece \piece{3}, your currently favorite piece, such that the trimmed piece will be equivalent to \piece{2}?".   Note that in this case (in contrast to the Equalize query), the trimmings may be more valuable than the trimmed piece.   When $Equalize^*$ queries are used, the agents are not allowed to choose the trimmings; the trimmings are kept for later iterations.  The agents are only allowed to take the trimmed pieces (hence, in contrast to the algorithms using Equalize, there is no guarantee on the proportionality of the allocation after a single run of the algorithm). Since the number of original pieces is $n$, all trimmed pieces must be taken. 

Based on the above observation, we now generalize a lemma proved by \citet{Aziz2015Discrete} from 4 to $n$ agents.

\subsection{Creating a single domination-edge from the VIP}

\begin{lemma}
Let $C$ be a cake and $X$ an envy-free division of a subset $C'\subset C$ among $n$ agents. Denote the remaining cake by $\overline{C'}=C\setminus C'$. Suppose that for two agents (e.g. Alice and Bob) the following holds:
\begin{align*}
V_A(X_A) - V_A(X_B) \geq V_A(\overline{C'}) / k
\end{align*}
where $k<n$. Then, after running an envy-free-VIP algorithm a bounded number $f(n)$ times with Alice as the VIP, Alice will dominate Bob.
\end{lemma}
\begin{proof}
Each run of an envy-free-VIP algorithm gives the VIP (Alice) a value of at least $1/n$. Hence, the value of the remaining cake decreases by a factor of at least $(n-1)/n$. Let $f(n)=\frac{\log n}{\log(n)-\log(n-1)}$. Note that $f(n)>\frac{\log k}{\log(n/(n-1))}$. Hence, after $f(n)$ iterations, the value of the remaining cake for Alice is at most $V_A(\overline{C'}) / k$. When this happens, the difference between Alice's value to Bob's value (in Alice's eyes) is more than the value of the remainder; hence Alice dominates Bob.
\end{proof}

Motivated by this lemma, we say that Alice \emph{$k$-dominates} Bob, if $V_A(X_A) - V_A(X_B) \geq V_A(\overline{C'}) / k$.   If Alice $k$-dominates Bob (where $k<n$), then after a number of steps which is a bounded function of $n$, Alice will dominate Bob. Hence, from now on, we add an edge in the domination graph whenever the source node $k$-dominates the target node for some $k<n$.

\begin{lemma}
\label{lem:one-domination-edge}
After a run of an envy-free-VIP algorithm for $n$ agents, the VIP $k$-dominates at least one other agent, where $k<n$.
\end{lemma}
\begin{proof}
An envy-free-VIP algorithm starts by the VIP agent (say, Alice) cutting the cake to $n$ equal pieces. Then, a certain number $k<n$ of pieces are trimmed. Consider the following two cases.

(a) $k=0$: all $n$ pieces are taken with no trimmings. Then, the division is fully envy-free and no cake is left, so domination is trivial.

(b) $1\leq k<n$: the divided cake is $C'\subseteq C$, and the remainder is $\overline{C'}=C\setminus C'$. This remainder is the union of the $k$ trimmings. Mark by \agentspiece{X}{i} the trimming taken from piece \piece{i}. Then:

$$\overline{C'} = \cup_{i=1}^k {\agentspiece{X}{i}}$$
By the additivity of Alice's value measure:

$$V_A(\overline{C'}) = \sum_{i=1}^k V_A({\agentspiece{X}{i}})$$
Assume, without loss of generality, that the trimming of piece \piece{1} has the largest value for Alice (Aziz and Mackenzie call such piece the \textbf{significant piece}). Then, by the pigeonhole principle, its value for Alice is at least $1/k$ the value of the remaining cake, so:

$$V_A(\agentspiece{X}{1}) \geq V_A(\overline{C'}) / k$$
This means that Alice $k$-dominates the agent that took piece \piece{1}.
\end{proof}

Lemma \ref{lem:one-domination-edge} guarantees that, after each run of an envy-free-VIP algorithm, the domination graph contains an edge going from the VIP agent to another agent. Hence, $n$ domination edges can be created by running the algorithm $n$ times with different VIP agents. But this may be insufficient to attain a solvable domination graph. The worst case is that, whenever a certain agent is the VIP, the same other agent takes the significant piece and hence the same domination edge is added again and again. Fortunately, \citet{Aziz2015Discrete} found a way to shift domination edges to other agents. 

\subsection{Creating two domination-edge from the VIP}
Suppose there are $n$ agents and Alice is the VIP. Suppose that after the first run, Alice dominates Bob. Our goal now is to make Alice dominate another agent. We run the algorithm again, this time keeping Bob as the last agent (the agent that does not trim). The other $n-2$ agents trim some of the pieces, until each of the first $n-1$ agents prefers at least two pieces (see Section \ref{sub:equalize} for a description on how it is done when $n=3$ and Section \ref{sec:4-agents} for the case $n=4$). Now, Bob has to choose a piece. Suppose w.l.o.g. that Bob's best piece is \piece{1} and his second-best piece is \piece{2}. There are two cases:

\textbf{Easy case}: \piece{1} is not the significant piece.  Then, another agent takes the significant piece and is $k$-dominated by Alice. Now two different domination edges emanate from Alice, as we wanted.

\textbf{Hard case}: \piece{1} is the significant piece.  Consider now what happens if Bob takes \piece{2} instead of \piece{1}. The other agents will not care, since each of the other agents prefers at least two pieces. But then Bob might envy the agent (say, Carl) who takes \piece{1}. In this case, we say that Bob \emph{competes} with Carl on the significant piece.  Let $\Delta V = V_B(\piece{1})-V_B(\piece{2})$. If Bob takes \piece{1} then Bob has an advantage of at least $\Delta V$ over Carl; if Bob takes \piece{2} then Bob has an envy of $\Delta V$ at Carl (Carl does not envy Bob in either case since Carl prefers two pieces).

Now, suppose the algorithm is run again and again with the same VIP, and each time we fall into the same hard case in which the same Bob prefers the significant piece. Eventually (after at most $n$ runs), Bob competes with an agent with whom he already competed in the past (e.g, Bob competes Carl again). Now, we ask Bob in which of these two runs the $\Delta V$ is larger. If the $\Delta V$ was larger in the first run, then in the second run we give Bob his second-best piece; if the $\Delta V$  was larger in the second run, then in the first run we change the allocation and give Bob his second-best piece. In either case, Bob will not be envious since the larger $\Delta V$ cancels the envy caused by the smaller $\Delta V$.

The above discussion can be summarized in the following lemma:

\begin{lemma}
\label{lem:two-domination-edges}
After at most $n$ runs of an envy-free-VIP algorithm for $n$ agents, the VIP $k$-dominates at least two other agents.
\end{lemma}

Plugging Lemma \ref{lem:two-domination-edges} into Lemma \ref{lem:solvable-n-2} yields Aziz \& Mackenzie's envy-free cake-cutting algorithm for 4 agents.

\begin{acks}
We are grateful for useful discussions and suggestions from several members of the Stack Exchange network  (http://stackexchange.com), particularly, Raphael Reitzig, Sebastian Wild, Abhishek Bansai, InstructedA, FrankW, Pal Gronas Drange (Pal GD), Anand S Kumar, babou, Bernard, BVMR, D.W., David Richerby, dmbarbour, Herbert, Jair Taylor, kcrisman, Luke Mathieson, Maroun Maroun, Rick Decker, Tsuyoshi Ito, Xoff and Yuval Filmus.

We are grateful to three anonymous PC members of the AAMAS 2015 conference and to two anonymous referees of the TALG journal for their helpful comments and suggestions.
\end{acks}

\bibliographystyle{ACM-Reference-Format-Journals}

\begin{thebibliography}{00}
	
	
	\ifx \showCODEN    \undefined \def \showCODEN     #1{\unskip}     \fi
	\ifx \showDOI      \undefined \def \showDOI       #1{{\tt DOI:}\penalty0{#1}\ }
	\fi
	\ifx \showISBNx    \undefined \def \showISBNx     #1{\unskip}     \fi
	\ifx \showISBNxiii \undefined \def \showISBNxiii  #1{\unskip}     \fi
	\ifx \showISSN     \undefined \def \showISSN      #1{\unskip}     \fi
	\ifx \showLCCN     \undefined \def \showLCCN      #1{\unskip}     \fi
	\ifx \shownote     \undefined \def \shownote      #1{#1}          \fi
	\ifx \showarticletitle \undefined \def \showarticletitle #1{#1}   \fi
	\ifx \showURL      \undefined \def \showURL       #1{#1}          \fi
	

	\bibitem[\protect\citeauthoryear{}{Amantidis
		et~al\mbox{.}}{2018}]%
	{Amantidis2018Improved}
	Georgios Amanatidis, George Christodoulou, John Fearnley, Evangelos Markakis, Christos-Alexandros Psomas, and Eftychia Vakaliou. 2018.
	\newblock \showarticletitle{{ An Improved Envy-Free Cake Cutting Protocol for Four Agents}}.
	
	
	\bibitem[\protect\citeauthoryear{Arzi, Aumann, and Dombb}{Arzi
		et~al\mbox{.}}{2011}]%
	{Arzi2011Throw}
	{Orit Arzi}, {Yonatan Aumann}, {and} {Yair Dombb}. 2016.
	\newblock \showarticletitle{{Toss one's cake, and eat it too: partial divisions can improve social welfare in cake cutting}}.
	\newblock {\em Social Choice and Welfare}  {46,} 4, 933--954.
	\newblock
	\showDOI{%
		\url{http://dx.doi.org/10.1007/s00355-015-0943-y}}
	
	\bibitem[\protect\citeauthoryear{Aziz}{Aziz}{2015}]%
	{Aziz2015Generalization}
	{Haris Aziz}. 2015.
	\newblock \showarticletitle{{A generalization of the AL method for fair
			allocation of indivisible objects}}.
	\newblock {\em Economic Theory Bulletin\/} (2015), 1--18.
	\newblock
	\showDOI{%
		\url{http://dx.doi.org/10.1007/s40505-015-0089-1}}
	
	
	\bibitem[\protect\citeauthoryear{Aziz and Mackenzie}{Aziz and
		Mackenzie}{2016a}]%
	{Aziz2015Discrete}
	{Haris Aziz} {and} {Simon Mackenzie}. 2016a.
	\newblock \showarticletitle{{A Discrete and Bounded Envy-free Cake Cutting
			Protocol for 4 Agents}}. In {\em Proceedings of STOC 2016}.
	\newblock
	\showDOI{%
		\url{http://dx.doi.org/10.1145/2897518.2897522}}
	
	
	\bibitem[\protect\citeauthoryear{Aziz and Mackenzie}{Aziz and
		Mackenzie}{2016b}]%
	{Aziz2016Discrete}
	{Haris Aziz} {and} {Simon Mackenzie}. 2016b.
	\newblock \showarticletitle{{A Discrete and Bounded Envy-Free Cake Cutting
			Protocol for Any Number of Agents}}. To appear in {\em Proceedings of FOCS 2016}.
	\newblock
	\showURL{%
		\url{http://arxiv.org/abs/1604.03655}}
	
	\bibitem[\protect\citeauthoryear{Barbanel and Brams}{Barbanel and Brams}{2004}]%
	{Barbanel2004Cake}
	{Julius Barbanel} {and} {Steven~J. Brams}. 2004.
	\newblock \showarticletitle{{Cake division with minimal cuts: envy-free procedures for three persons, four persons, and beyond}}.
	\newblock {\em Mathematical Social Sciences\/} {48,} 3, 251--269.

	
	\bibitem[\protect\citeauthoryear{Brams, Kilgour, and Klamler}{Brams
		et~al\mbox{.}}{2013}]%
	{Brams2013TwoPerson}
	{Steven~J. Brams}, {D.~Marc Kilgour}, {and} {Christian Klamler}. 2013.
	\newblock \showarticletitle{{Two-Person Fair Division of Indivisible Items: An
			Efficient, Envy-Free Algorithm}}.
	\newblock {\em Social Science Research Network Working Paper Series\/} (5 June
	2013), 130--141.
	\newblock
	\showDOI{%
		\url{http://dx.doi.org/10.1090/noti1075}}
	
	\bibitem[\protect\citeauthoryear{Brams and Taylor}{Brams and Taylor}{1995}]%
	{Brams1995EnvyFree}
	{Steven~J. Brams} {and} {Alan~D. Taylor}. 1995.
	\newblock \showarticletitle{{An Envy-Free Cake Division Protocol}}.
	\newblock {\em The American Mathematical Monthly\/} {102}, 1 (Jan. 1995), 9--18.
	\newblock
	\showISSN{00029890}
	\showDOI{%
		\url{http://dx.doi.org/10.2307/2974850}}
	
	
	\bibitem[\protect\citeauthoryear{Brams and Taylor}{Brams and Taylor}{1996}]%
	{Brams1996Fair}
	{Steven~J. Brams} {and} {Alan~D. Taylor}. 1996.
	\newblock {\em {Fair Division: From Cake-Cutting to Dispute Resolution}}.
	\newblock Cambridge University Press.
	\newblock
	\showISBNx{0521556449}
	\showURL{%
		\url{http://www.worldcat.org/isbn/0521556449}}
	
	
	\bibitem[\protect\citeauthoryear{Br\^{a}nzei}{Br\^{a}nzei}{2015}]%
	{Branzei2015Note}
	{Simina Br\^{a}nzei}. 2015.
	\newblock \showarticletitle{{A note on envy-free cake cutting with polynomial
			valuations}}.
	\newblock {\it Information Processing Letters} {115}, 2 (Feb. 2015), 93--95.
	\newblock
	\showISSN{00200190}
	\showDOI{%
		\url{http://dx.doi.org/10.1016/j.ipl.2014.07.005}}
	
	
	\bibitem[\protect\citeauthoryear{Br\^{a}nzei, Caragiannis, Kurokawa, and
		Procaccia}{Br\^{a}nzei et~al\mbox{.}}{2016}]%
	{Branzei2016Algorithmic}
	{Simina Br\^{a}nzei}, {Ioannis Caragiannis}, {David Kurokawa}, {and} {Ariel~D.
		Procaccia}. 2016.
	\newblock \showarticletitle{{An Algorithmic Framework for Strategic Fair
			Division}}.
	In {\em
		Proceedings of the 30th AAAI Conference on Artificial
		Intelligence}. 
	\newblock
	\showURL{%
		\url{http://arxiv.org/abs/1307.2225}}
	
	
	\bibitem[\protect\citeauthoryear{Deng, Qi, and Saberi}{Deng
		et~al\mbox{.}}{2012}]%
	{Deng2012Algorithmic}
	{Xiaotie Deng}, {Qi Qi}, {and} {Amin Saberi}. 2012.
	\newblock \showarticletitle{{Algorithmic Solutions for Envy-Free Cake
			Cutting}}.
	\newblock {\em Operations Research\/} {60}, 6 (Dec. 2012), 1461--1476.
	\newblock
	\showISSN{1526-5463}
	\showDOI{%
		\url{http://dx.doi.org/10.1287/opre.1120.1116}}
	
	
	\bibitem[\protect\citeauthoryear{Developers}{Developers}{2015}]%
	{sage}
	{The~Sage Developers}. 2015.
	\newblock {\em {Sage Mathematics Software, Version 6.7}}.
	\newblock
	\newblock
	\shownote{http://www.sagemath.org.}
	
	
	\bibitem[\protect\citeauthoryear{Edmonds and Pruhs}{Edmonds and Pruhs}{2006}]%
	{Edmonds2006Balanced}
	{Jeff Edmonds} {and} {Kirk Pruhs}. 2006.
	\newblock \showarticletitle{{Balanced Allocations of Cake}}. In {\em Proceedings of FOCS}, 	Vol.~47. IEEE Computer Society, 623--634.
	\newblock
	\showISBNx{0-7695-2720-5}
	\showISSN{0272-5428}
	\showDOI{%
		\url{http://dx.doi.org/10.1109/focs.2006.17}}
	
	
	\bibitem[\protect\citeauthoryear{Edmonds and Pruhs}{Edmonds and Pruhs}{2011}]%
	{Edmonds2011Cake}
	{Jeff Edmonds} {and} {Kirk Pruhs}. 2011.
	\newblock \showarticletitle{{Cake cutting really is not a piece of cake}}.
	\newblock {\em ACM Transactions on Algorithms\/} {7}, 4 (1 Sept. 2011), 1--12.
	\newblock
	\showISSN{1549-6325}
	\showDOI{%
		\url{http://dx.doi.org/10.1145/2000807.2000819}}
	
	
	\bibitem[\protect\citeauthoryear{Edmonds, Pruhs, and Solanki}{Edmonds
		et~al\mbox{.}}{2008}]%
	{Edmonds2008Confidently}
	{Jeff Edmonds}, {Kirk Pruhs}, {and} {Jaisingh Solanki}. 2008.
	\newblock \showarticletitle{{Confidently Cutting a Cake into Approximately Fair
			Pieces}}.
	\newblock {\em Algorithmic Aspects in Information and Management\/} (2008),
	155--164.
	\newblock
	\showDOI{%
		\url{http://dx.doi.org/10.1007/978-3-540-68880-8\_16}}
	
	
	\bibitem[\protect\citeauthoryear{Even and Paz}{Even and Paz}{1984}]%
	{Even1984Note}
	{Shimon Even} {and} {Azaria Paz}. 1984.
	\newblock \showarticletitle{{A note on cake cutting}}.
	\newblock {\em Discrete Applied Mathematics\/} {7}, 3 (March 1984), 285--296.
	\newblock
	\showISSN{0166218X}
	\showDOI{%
		\url{http://dx.doi.org/10.1016/0166-218x(84)90005-2}}
	
	\bibitem[\protect\citeauthoryear{Kurokawa, Lai, and Procaccia}{Kurokawa
		et~al\mbox{.}}{2013}]%
	{Kurokawa2013How}
	{David Kurokawa}, {John~K. Lai}, {and} {Ariel~D. Procaccia}. 2013.
	\newblock \showarticletitle{{How to Cut a Cake Before the Party Ends}}. 
		In {\em
		Proceedings of the 27th AAAI Conference on Artificial
		Intelligence}. 555--561.
	\newblock
	
	\bibitem[\protect\citeauthoryear{Pikhurko}{Pikhurko}{2000}]%
	{Pikhurko2000EnvyFree}
	{Oleg Pikhurko}. 2000.
	\newblock \showarticletitle{{On Envy-Free Cake Division}}.
	\newblock {\it American Mathematical Monthly}  {107} (2000), 736--738.
	\newblock
	\showDOI{%
		\url{http://dx.doi.org/10.2307/2695471}}
	
	
	\bibitem[\protect\citeauthoryear{Procaccia}{Procaccia}{2009}]%
	{Procaccia2009Thou}
	{Ariel~D. Procaccia}. 2009.
	\newblock \showarticletitle{{Thou Shalt Covet Thy Neighbor's Cake}}. In {\em
		Proceedings of the 21st International Joint Conference on Artificial
		Intelligence}. 239--244.
	\newblock

	
	\bibitem[\protect\citeauthoryear{Procaccia}{Procaccia}{2015}]%
	{Procaccia2015Cake}
	{Ariel~D. Procaccia}. 2015.
	\newblock \showarticletitle{{Cake Cutting Algorithms}}.
	\newblock In {\em Handbook of Computational Social Choice}, {Felix Brandt},
	{Vincent Conitzer}, {Ulle Endriss}, {Jerome Lang}, {and} {Ariel~D. Procaccia}
	(Eds.). Cambridge University Press, Chapter~13.
	\newblock
	
	\bibitem[\protect\citeauthoryear{Reitzig and Wild}{Reitzig and Wild}{2017}]%
	{Reitzig2015Efficient}
	{Raphael Reitzig} {and} {Sebastian Wild}. 2017.
	\newblock {Building Fences Straight and High: An Optimal Algorithm for Finding the Maximum Length You Can Cut k Times from Given Sticks}.
	\newblock In {\em Algorithmica	}
	\newblock (2017).
	\newblock
	\shownote{arXiv preprint 1502.04048.}
	
	
	\bibitem[\protect\citeauthoryear{Robertson and Webb}{Robertson and
		Webb}{1998}]%
	{Robertson1998CakeCutting}
	{Jack~M. Robertson} {and} {William~A. Webb}. 1998.
	\newblock {\em {Cake-Cutting Algorithms: Be Fair if You Can}\/} (first ed.).
	\newblock A K Peters/CRC Press. 1--177 pages.
	\newblock
	\showISBNx{1568810768}
	\showURL{%
		\url{http://www.worldcat.org/isbn/1568810768}}
	
	
	\bibitem[\protect\citeauthoryear{Saberi and Wang}{Saberi and Wang}{2009}]%
	{Saberi2009Cutting}
	{Amin Saberi} {and} {Ying Wang}. 2009.
	\newblock \showarticletitle{{Cutting a Cake for Five People}}.
	\newblock In {\em Algorithmic Aspects in Information and Management},
	{Andrew~V. Goldberg} {and} {Yunhong Zhou} (Eds.). Lecture Notes in Computer
	Science, Vol. 5564. Springer Berlin Heidelberg, 292--300.
	\newblock
	\showDOI{%
		\url{http://dx.doi.org/10.1007/978-3-642-02158-9\_25}}
	
	
	\bibitem[\protect\citeauthoryear{Segal-Halevi, Hassidim, and
		Aumann}{Segal-Halevi et~al\mbox{.}}{2015a}]%
	{SegalHalevi2015EnvyFree}
	{Erel Segal-Halevi}, {Avinatan Hassidim}, {and} {Yonatan Aumann}. 2015a.
	\newblock \showarticletitle{{Envy-Free Cake-Cutting in Two Dimensions}}. In
	{\em Proceedings of the 29th AAAI Conference on Artificial Intelligence}. 1021--1028.
	\newblock
	
	
	\bibitem[\protect\citeauthoryear{Segal-Halevi, Hassidim, and
		Aumann}{Segal-Halevi et~al\mbox{.}}{2015b}]%
	{SegalHalevi2015Waste}
	{Erel Segal-Halevi}, {Avinatan Hassidim}, {and} {Yonatan Aumann}. 2015b.
	\newblock \showarticletitle{{Waste Makes Haste: Bounded Time Protocols for
			Envy-Free Cake Cutting with Free Disposal}}. In {\em Proceedings of the 14th International Conference on Autonomous Agents and Multi-Agent Systems}. 901--908.
	\newblock
	\showDOI{%
		\url{http://dx.doi.org/10.13140/RG.2.1.1902.0645}}
	
	
	\bibitem[\protect\citeauthoryear{Steinhaus}{Steinhaus}{1948}]%
	{Steinhaus1948Problem}
	{Hugo Steinhaus}. 1948.
	\newblock \showarticletitle{{The problem of fair division}}.
	\newblock {\em Econometrica\/} {16}, 1 (Jan. 1948), 101--104.
	\newblock
	\showURL{%
		\url{http://www.jstor.org/stable/1914289}}
	
	
	\bibitem[\protect\citeauthoryear{Stromquist}{Stromquist}{1980}]%
	{Stromquist1980How}
	{Walter Stromquist}. 1980.
	\newblock \showarticletitle{{How to Cut a Cake Fairly}}.
	\newblock {\em The American Mathematical Monthly\/} {87}, 8 (Oct. 1980),
	640--644.
	\newblock
	\showISSN{00029890}
	\showDOI{%
		\url{http://dx.doi.org/10.2307/2320951}}
	
	
	\bibitem[\protect\citeauthoryear{Stromquist}{Stromquist}{2008}]%
	{Stromquist2008Envyfree}
	{Walter Stromquist}. 2008.
	\newblock \showarticletitle{{Envy-free cake divisions cannot be found by finite
			protocols}}.
	\newblock {\em Electronic Journal of Combinatorics\/} (Jan. 2008).
	\newblock
	\showURL{%
		\url{http://www.emis.ams.org/journals/EJC/Volume\_15/PDF/v15i1r11.pdf}}
	
	
	\bibitem[\protect\citeauthoryear{Su}{Su}{1999}]%
	{Su1999Rental}
	{Francis~E. Su}. 1999.
	\newblock \showarticletitle{{Rental Harmony: Sperner's Lemma in Fair
			Division}}.
	\newblock {\em The American Mathematical Monthly\/} {106}, 10 (Dec. 1999),
	930--942.
	\newblock
	\showISSN{00029890}
	\showDOI{%
		\url{http://dx.doi.org/10.2307/2589747}}
	
	
	\bibitem[\protect\citeauthoryear{Woeginger and Sgall}{Woeginger and
		Sgall}{2007}]%
	{Woeginger2007Complexity}
	{Gerhard~J. Woeginger} {and} {Ji\v{r}\'{\i} Sgall}. 2007.
	\newblock \showarticletitle{{On the complexity of cake cutting}}.
	\newblock {\em Discrete Optimization\/} {4}, 2 (June 2007), 213--220.
	\newblock
	\showISSN{15725286}
	\showDOI{%
		\url{http://dx.doi.org/10.1016/j.disopt.2006.07.003}}
\end{thebibliography}


\end{document}